\documentclass[11pt,reqno]{amsart}
\usepackage[utf8]{inputenc}
\usepackage{amssymb,mathrsfs,xcolor,bbold,comment,enumerate}
\usepackage[colorlinks=true,urlcolor=blue,citecolor=purple,linkcolor=blue]{hyperref}
\numberwithin{equation}{section}
\usepackage[all]{xy}
\usepackage{color}
\usepackage{amsaddr}
\usepackage{apptools}
\usepackage[onehalfspacing]{setspace}
\AtAppendix{\counterwithin{theorem}{section}}
\AtAppendix{\counterwithin{equation}{section}}

\makeatletter
\renewcommand{\email}[2][]{
\ifx\emails\@empty\relax\else{\g@addto@macro\emails{,\space}}\fi
\@ifnotempty{#1}{\g@addto@macro\emails{\textrm{(#1)}\space}}
\g@addto@macro\emails{#2}
}
\makeatother

\theoremstyle{plain}
\newtheorem{theorem}{Theorem}[section]

\newtheorem*{theorem*}{Theorem}
\newtheorem{proposition}[theorem]{Proposition}
\newtheorem{corollary}[theorem]{Corollary}
\newtheorem{lemma}[theorem]{Lemma}

\theoremstyle{definition}
\newtheorem{assumption}[theorem]{Assumption}
\newtheorem{remark}[theorem]{Remark}

\newtheorem{example}[theorem]{Example}
\newtheorem{definition}[theorem]{Definition}

\newcommand{\mf}{\mathfrak}

\DeclareMathOperator{\es}{\rm ES}

\newcommand{\R}{\mathbb{R}} 
\newcommand{\Q}{\mathbb{Q}}

\newcommand{\E}{\mathbb{E}}
\newcommand{\tn}{\textnormal}

\newcommand{\ind}{\mathbf{1}}

\renewcommand{\P}{\mathbb{P}}
\newcommand{\Norm}{\|\cdot\|}
\newcommand{\N}{\mathbb{N}}
\newcommand{\dom}{\textnormal{dom}}

\newcommand{\CN}{\mathcal N}
\newcommand{\CF}{\mathcal F}
\newcommand{\CA}{\mathcal A}
\newcommand{\CG}{\mathcal G}

\newcommand{\CX}{\mathcal X}
\newcommand{\CP}{\mathcal P}

\newcommand{\CI}{\mathcal I}

\renewcommand{\Q}{{\mathbb Q}}

\newcommand{\CH}{\mathcal H}
\newcommand{\Linfty}{L^\infty}
\newcommand{\ph}{\varphi}

\newcommand{\price}{\mathfrak p}

\newcommand{\peq}{\preceq}

\newcommand{\mbf}{\mathbf}

\newcommand{\w}{\widehat}

\newcommand{\Scal}{{\mathcal S}}

\title[Risk sharing under heterogeneous beliefs without convexity]{Risk sharing under heterogeneous beliefs without convexity}

\author[F.-B.~Liebrich]{Felix-Benedikt Liebrich}
\address{Institute of Actuarial and Financial Mathematics \& House of Insurance,\\
Leibniz Universit\"at Hannover, Germany}
\email{felix.liebrich@insurance.uni-hannover.de}
\thanks{\textit{Acknowledgements:} I would like to thank two anonymous referees, Cosimo Munari, Gregor Svindland, Ruodu Wang, and the participants of the Seminar on Risk Management and Actuarial Science at the University of Waterloo for valuable comments and discussions related to this work.}
\date{May 2, 2022}

\begin{document}

\parindent 0em \noindent

\begin{abstract}
We consider the problem of finding Pareto-optimal allocations of risk among finitely many agents. 
The associated individual risk measures are law invariant, but with respect to agent-dependent and potentially heterogeneous reference probability measures. Moreover, we assume that the individual risk assessments are consistent with the respective second-order stochastic dominance relations. We do \textit{not} assume their convexity though. A simple sufficient condition for the existence of Pareto optima is provided. The proof combines local comonotone improvement with a Dieudonn\'e-type argument, which also establishes a link of the optimal allocation problem to the realm of ``collapse to the mean" results. Finally, we extend the results to capital requirements with multidimensional security markets. 
\end{abstract}

\maketitle


\section{Introduction}

This paper addresses the problem of finding optimal allocations of risk among finitely many agents, i.e., optimisers for the problem
\begin{equation}\begin{array}{l}\label{problem}\sum_{i\in\CI}\rho_i(X_i)\longrightarrow\min\\
\text{subject to }\sum_{i\in\CI}X_i=X.\end{array}\end{equation}
The agents under consideration form a finite set $\mathcal I$. Each agent $i\in\CI$ measures the risk of net losses $Y$ in a space $\CX$ with a risk measure $\rho_i$. Given the total loss $X$ collected in the system, an allocation attributes a portion $X_i\in\CX$ to each agent, i.e., the condition $\sum_{i\in\CI}X_i=X$ holds. We do not impose any restriction on the notion of an allocation, i.e., every vector whose coordinates sum up to $X$ is hypothetically feasible.
An allocation of $X$ is {\em optimal} (or Pareto optimal) if it minimises the aggregated risk $\sum_{i\in\CI}\rho_i(X_i)$ in the system.
This risk-sharing problem is of fundamental importance for the theory of risk measures, for capital allocations in capital adequacy, and in the design and discussion of regulatory frameworks; cf.~\cite{FilKup,Tsanakas,Weber}. 

Usually, the space of net losses appearing in \eqref{problem} is an infinite-dimensional space of \textit{random variables}, which complicates finding optima. 
Individual risk measures are often assumed to be convex and monetary, standard axioms in the theory of risk measures since \cite{FoeSch2002,Frittelli} and also common for utility assessments such as  variational preferences; cf.\ \cite{Variational}.
Convexity alone is not sufficient though. A powerful solution theory has so far mostly been established under the additional assumption of \textit{law invariance}, meaning that the risk of a random variable $X$ on a probability space $(\Omega,\CF,\P)$ merely depends on its distribution under the reference measure $\P$. 
A rich strand of literature has studied the wide-ranging analytic consequences of law invariance in conjunction with convexity; see, e.g., \cite{General,Fatou,FilSvi,JST2,Leung,Allocations,continuity}, and the references therein. Studies of the risk sharing problem for convex monetary risk measures (or equivalently, concave monetary utility functions) are~\cite{Acciaio_alone,ElKaroui,Burgert,Svindland,JST}. For risk sharing problems with special  law-invariant, but not necessarily convex functionals, see
~\cite{Embrechts1,Liuetal}.

The solution theory in the law-invariant case splits in two steps. In a first step, the set of optimisation-relevant allocations is reduced to \textit{comonotone} allocations. This is achieved by \textit{comonotone improvement results} based on~\cite{Landsberger,Carlier,Svindland,Rueschendorf}. 
In a second step, one finds suitable bounds which prove that relevant comonotone allocations form a compact set. That enables approximation procedures and the selection of converging optimising sequences.\footnote{~These two steps are separated explicitly in, e.g., \cite{Allocations}, where also numerous economic optimisation problems are analysed according to that scheme. Under cash-additivity of the involved functionals, the second step usually poses no problem as one can ``rebalance cash" among the involved agents. For the details, see for instance \cite{Svindland}.} 

The previous arguments tend to rely fundamentally on the existence of a single ``objective" reference probability measure which is shared and accepted by all agents in $\CI$. A more recent and quickly growing strand of literature in finance, insurance, and economics, however, dispenses with this paradigm. Two economic considerations motivate these {\em heterogeneous} reference probability measures. Firstly, different agents may have access to different sources of information, resulting in information asymmetry and in different subjective beliefs. Secondly, agents may entertain heterogeneous probabilistic subjective beliefs as a result of their \textit{preferences} or their use of different internal models.\footnote{~While heterogeneous subjective beliefs and information asymmetry may not be synonymous, the distinction will not play a role for the purposes of the present paper.}

Applying heterogeneous (probabilistic) beliefs in risk sharing and related problems, \cite{Amarante} explores the demand for insurance when the insurer exhibits ambiguity, whereas the insured is an expected-utility agent.
Under heterogeneous reference probabilities for insurer and insured, \cite{Boonen} provides optimal reinsurance designs, and \cite{Chi,Ghossoub} generalise Arrow's ``Theorem of the Deductible''.
\cite{Boonenetal} studies bilateral risk sharing with exposure constraints and admit a very general relation between the two involved reference probabilities. \cite[Section 6]{Asimitetal} consider optimal risk sharing under heterogeneous reference probabilities affected by exogenous triggering events, while \cite{LeVan} studies the existence of Pareto optima and equilibria in a finite-dimensional setting when optimality of allocations is assessed relative to individual and potentially heterogeneous sets of probabilistic priors.
\cite{Embrechts2} studies risk sharing for Value-at-Risk and Expected Shortfall agents endowed with heterogeneous reference probability measures, and provide explicit formulae for the shape of optimal allocations.
\cite{Liu} makes similar contributions to weighted risk sharing with distortion risk measures under heterogeneous beliefs.
Finally, \cite{Acciaio} is devoted the existence question for optimal allocations among two agents with heterogeneous reference models, each measuring risk with a law-invariant concave monetary utility function.
The risk sharing problem is constrained to random variables over a \textit{finite} $\sigma$-algebra though, which reduces the optimisation problem to a finite-dimensional space. Moreover, one of the two reference probabilities involved only takes rational values on the finite $\sigma$-algebra in question. 

\vspace{0.2cm}

Against the outlined backdrop of the existing literature, several key features of our results stand out. 

\vspace{0.2cm}

\textsf{Infinite-dimensional setting:} Whereas dealing with heterogeneous reference probability has sometimes been facilitated by restriction to a finite dimensional setting, we consider problem~\eqref{problem} for the space $\CX$ of all bounded random variables over an atomless probability space $(\Omega,\CF,\P)$, a \textit{bona fide} infinite-dimensional space.

\vspace{0.2cm}

\textsf{True generalisation of most known results:} As already stated, we consider a finite set $\CI$ of agents, each measuring risk with a functional $\rho_i\colon L^\infty\to\R$ which is law invariant with respect to a probability measure $\Q_i$ equivalent to $\P$.
The latter only plays the role of a gauge and we are free to assume $\P=\Q_j$ for some $j\in\CI$.
Our main results in Section~\ref{sec:main} provide mild sufficient conditions for the existence of optimal risk allocations without a rationality condition on involved probabilities as in \cite{Acciaio}. 
Mathematically, key to our strategy is to adapt the two-step procedure of comonotone improvement to the case of heterogeneous reference probabilities. 
This will mostly be achieved in Appendix~\ref{sec:local}.
However, ``rebalancing cash" is not an option anymore under heterogeneity. 
Instead, the second step of the two-step procedure sketched above will be crucial. While optimal allocations will usually \textit{not} be comonotone in the heterogeneous case, 
our procedure nevertheless unifies various existence results for optimal allocations under a single reference probability measure and extends them to the case of multiple reference measures. 

\vspace{0.2cm} 

\textsf{Nonconvex risk measures:} A third and particularly important aspect is that we do not assume convexity (or concavity) of the involved functionals. Instead, we work with three axiomatic properties of risk measures introduced and studied in the very recent literature as alternatives to (quasi)convexity. 
While this trend is in la-rge parts motivated by the lacking convexity of many distortion risk measures such as the Value-at-Risk, \cite[Section 1.2]{AmaranteRepr} presents further critical remarks on subadditivity.

A requirement we shall impose throughout our study is the assumption that individual risk measures be {\em consistent}, a property recently axiomatised by \cite{Consistent}.\footnote{~\cite[Section 4]{Consistent} already provides a preliminary study of risk sharing with consistent risk measures under a single reference probability measure; see also Section~\ref{sec:discussion}.}
While many desirable characterisations are shown to be equivalent, their  eponymous characteristic feature is monotonicity with respect to second-order stochastic dominance. Each law-invariant and convex monetary risk measure is consistent (up to an affine transformation), but the converse implication does \textit{not} hold. 

The second ingredient---which is mostly of technical relevance and makes the assumptions in the main results particularly satisfiable---is {\em star-shapedness}. Star-shaped risk measures have been studied systematically in the recent working paper~\cite{Starshape}. They are motivated, for instance, by the observation that subadditive risk measures intertwine the measurement of concentration and liquidity effects (by curves $\R_+\ni t\mapsto\rho(tX)$) with diversification benefits from {\em merging} portfolios. The latter is translated by (quasi)convexity of the risk measures in question.
Star-shaped risk measures are more agnostic and replace (quasi)convexity by the demand that decreasing exposition to an acceptable loss profile (having at most neutral risk) does not lead to a loss of acceptability:
\[\forall\, X\in\CX\,\forall\,\lambda\in[0,1]:\quad \rho(X)\le 0\quad\Longrightarrow\quad \rho(\lambda X)\le 0.\]
Further economic motivation is discussed in \cite[Section 2]{Starshape}. 

Third, we shall require a certain compatibility of probabilistic beliefs and assume the existence of a finite measurable partition $\pi$ of $\Omega$ such that the agents agree on the associated conditional distributions. 
This is akin to and generalises the (much more specific) setting of \cite{Marshall}, one of the earliest contributions on heterogeneous reference measures. 
While Section~\ref{sec:heterogeneous} will give more details and shed more light on this assumption, 
we anticipate here the consequence that individual risk measures  fall in the most relevant class of {\em scenario-based risk measures} recently introduced in~\cite{WangZiegel}. 
A desirable aspect of the latter perspective is that scenario-basedness is preserved under the infimal convolution operation (cf.\ Corollary~\ref{cor:Pbased}), while the individual law-invariance is lost in the general heterogeneous case. 

\vspace{0.2cm}

\textsf{Dispensing with law invariance:} Last, we will also drop law invariance of the involved functionals {\em ex post}. In Theorem~\ref{thm:application2}, we consider \textit{capital requirements} computed on the basis of the acceptance set of a consistent risk measure and a potentially multidimensional frictionless security market. 
The former poses a capital adequacy test that the agent attempts to pass with hedging instruments from the security market. 
These capital requirements follow the spirit of~\cite{FritScand} and are studied in, e.g., \cite{Baes,FKM2015,Liebrich}. 
Depending on the securities available in the security market, these functionals may well not be cash-additive. 
We will show that in this case, they also lose the property of being law invariant.
In contrast to most existence results for optimal risk allocations in the literature, our approach is able to cope with the observed lack of law invariance and yields a simple sufficient condition for the existence of Pareto-optimal allocations.

\vspace{0.2cm}

The paper unfolds as follows. Section~\ref{sec:prelim} collects preliminaries. Consistent and star-shaped risk measures and their admissibility for our main results are studied at length in Section~\ref{sec:consistent}. 
In Section~\ref{sec:heterogeneous}, we carefully introduce our setting of heterogeneous reference probability measures and motivate the buttressing assumption from four different angles. 
The main results are then collected in Section~\ref{sec:main}. Theorem~\ref{thm:risksh_mon} covers the case of agents endowed with consistent risk measures. The latter are cash-additive and law invariant by definition. Theorem~\ref{thm:application2}, goes one step further and addresses capital requirements with multidimensional security markets and acceptance sets induced by consistent risk measures, thus dispensing with cash-additivity and law invariance. 
All proofs, mathematical details, and auxiliary results can be found in Appendices~\ref{appendix}--\ref{last appendix}. Examples illustrating the main text are relegated to Appendix~\ref{sec:examples}.

\subsection{Preliminaries}\label{sec:prelim}

We first outline terminology, notation, and conventions adopted throughout the paper:
\begin{itemize}
\item The effective domain of a function $f\colon S\to[-\infty,\infty]$ defined on a nonempty set $S$ is 
\begin{center}$\dom(f):=\{s\in S\mid f(s)\in\R\}$.\end{center}

\item For an arbitrary natural number $K\in\N$, we denote the set $\{1,\ldots,K\}$ by $[K]$.

\item Bold-faced symbols denote vectors of objects. 

\item $(\Omega,\CF,\P)$ is an \textit{atomless} probability space, i.e., it admits a random variable whose cumulative distribution function under $\P$ is continuous. In particular, $\P$ has {\em convex range}, i.e., for all $A\in\CF$ and all $p\in[0,\P(A)]$ we find $B\in\CF$ with the properties $B\subset A$ and $\P(B)=p$.
Whenever $\Q\ll\P$ is another probability measure, $\Q$ is also atomless. 
\item $L^0$ denotes the space of equivalence classes up to $\P$-almost-sure ($\P$-a.s.) equality of real-valued random variables. The subspaces of equivalence classes of bounded and $\P$-integrable random variables are denoted by $L^\infty$ and $L^1$, respectively. 
All these spaces are canonically equipped with the $\P$-a.s.\ order $\le$, and all appearing (in)equalities between random variables are to be understood in this sense. We denote the respective positive cones by $L^\infty_+$ and $L^1_+$, and the supremum norm on $L^\infty$ by $\Norm_\infty$.

\item If we consider the spaces $L^\infty$ and $L^1$ with respect to a probability measure $\Q\neq\P$ on $(\Omega,\CF)$, we shall write $L^\infty_\Q$ and $L^1_\Q$. 

\item Expectations and conditional expectations, given a sub-$\sigma$-algebra $\CG\subset\CF$, computed with respect to a measure $\Q$ are denoted by $\E_\Q[\cdot]$ and $\E_\Q[\cdot|\CG]$. 

\item For a probability measure $\Q$ on $(\Omega,\CF)$ and an event $B\in\CF$ with $\Q(B)>0$, we define the conditional probability measure $\Q^B\colon\CF\to[0,1]$ by $\Q_B(A):=\tfrac{\Q(A\cap B)}{\Q(B)}$. 

\item The absolute continuity relation between two probability measures $\P$ and $\Q$ on $(\Omega,\CF)$ is denoted by $\Q\ll\P$, and equivalence of probability measures by $\Q\approx\P$.

\item If for two elements $X,Y\in L^0$ and a probability measure $\Q\ll\P$ the distributions $\Q\circ Y^{-1}$ of $Y$ under $\Q$ agrees with the distribution $\Q\circ X^{-1}$ of $X$ under $\Q$, we shall write $X\sim_\Q Y$. 

\item A subset $\CA\subset L^0$ is \emph{law invariant} with respect to a probability measure $\Q\ll\P$ (or $\Q$-law invariant) if, whenever $X\in\CA$ and $Y\sim_\Q X$, then also $Y\in\CA$.

\item Given nonempty sets $\CA\subset L^0$ and $S$, a function $f\colon\CA\to S$ is law invariant with respect to a probability measure $\Q\ll\P$ if $f(X)=f(Y)$ holds for all $X,Y\in\CA$ satisfying $X\sim_\Q Y$.

\item $\Pi$ denotes the set of finite measurable partitions $\pi\subset\CF$ of $\Omega$ satisfying $\P(B)>0$ for all $B\in\pi$.

\item Denoting the space of all bounded linear functionals on $L^\infty$ by $(L^\infty)^*$, the \emph{convex conjugate} of a (not necessarily convex) function $f\colon L^\infty\to[-\infty,\infty]$ is the function $f^*\colon(L^\infty)^*\to[-\infty,\infty]$ defined by 
\[f^*(\phi)=\sup\{\phi(X)-f(X)\mid X\in L^\infty\}.\]
\item A \emph{monetary risk measure} $\rho\colon L^\infty\to\R$ is a map that is:
\begin{enumerate}[(a)]
\item \emph{monotone}, i.e., for $X,Y\in L^\infty$ with $X\le Y$, $\rho(X)\le \rho(Y)$ holds.
\item \emph{cash-additive}, i.e., 
\[\forall\,X\in L^\infty\,\forall\,m\in\R:\quad \rho(X+m)=\rho(X)+m.\]
\end{enumerate}
Its {\em acceptance set} $\CA_\rho$ collects all loss profiles that bear neutral risk: 
\[\CA_\rho:=\{X\in\CX\mid \rho(X)\le 0\}.\]
By monotonicity and cash-additivity, $\rho$ is a norm-continuous function and the acceptance set $\CA_\rho$ is closed.
$\rho$ can also be recovered from $\CA_\rho$ with the formula
\begin{center}$\rho(X)=\inf\{m\in\R\mid X-m\in\CA_\rho\},\quad X\in L^\infty.$\end{center}
\noindent As mentioned in the introduction and reflected by the preceding definition, risk measures are applied to losses net of gains in this manuscript, not gains net of losses. In particular, nonnegative random variables correspond to pure losses. 
\item A monetary risk measure $\rho\colon L^\infty\to\R$ is {\em normalised} if $\rho(0)=0$.
\item The {\em asymptotic cone} of the acceptance set $\CA_\rho$ of a risk measure $\rho$ is the set $\CA_\rho^\infty$ of all $U\in L^\infty$ that can be represented as $U=\lim_{k\to\infty}s_kY_k$ for a vanishing sequence $(s_k)_{k\in\N}\subset(0,\infty)$ and a sequence $(Y_k)_{k\in\N}\subset\CA_\rho$.
More information on asymptotic cones in a finite-dimensional setting can be found in \cite[Chapter 2]{Auslender}. \end{itemize}

\section{Admissible risk measures}\label{sec:consistent}

\subsection{Consistent and star-shaped risk measures}

 This paper proves the existence of optimal risk sharing schemes for {\em consistent} risk measures introduced in \cite{Consistent}. Consistency means that the risk assessment respects the second-order stochastic dominance relation between arguments. Given random variables $X,Y\in L^\infty$ and a probability measure $\Q\ll\P$, recall that $Y$ dominates $X$ in $\Q$-second-order stochastic dominance relation if $\E_\Q[v(X)]\le\E_\Q[v(Y)]$ holds for all convex and nondecreasing functions $v\colon\R\to\R$.\footnote{~In view of the references we cite, we shall mostly work with the convex order defined in Appendix~\ref{proofs of sec 3}. In fact, if $Y$ dominates $X$ with respect to the $\Q$-convex order, $Y$ also second-order stochastically dominates $X$.}  

\begin{definition}[cf.\ \cite{Consistent}]\label{def:consistent}
Let $\Q\ll\P$ be a probability measure. A normalised monetary risk measure $\rho\colon L^\infty\to\R$ is a \emph{$\Q$-consistent risk measure} if, whenever $Y\in L^\infty$ dominates $X\in L^\infty$ in second-order stochastic dominance relation under $\Q$, then also 
$\rho(X)\le\rho(Y)$.
\end{definition}

Each normalised, convex, and $\Q$-law-invariant monetary risk measure is $\Q$-consistent, but also minima of such risk measures (cf.\ \cite[Proposition 3.2]{Consistent}).
More precisely, by \cite[Theorem 3.3]{Consistent}, for each consistent risk measure $\rho$ there is a set $\mathcal T$ of $\Q$-law-invariant and convex monetary risk measures $\tau$ such that
\begin{equation}\label{eq:consistent}\rho(X)=\min_{\tau\in\mathcal T}\tau(X),\quad X\in L^\infty.\end{equation}
A direct consequence of \eqref{eq:consistent} is the following formula for the convex conjugate $\rho^*$ of $\rho$:
\begin{equation}\label{eq:conjugate}\rho^*(\phi)=\sup_{Y\in\CA_\rho}\phi(Y)=\sup_{\tau\in\mathcal T}\tau^*(\phi),\quad \phi\in(L^\infty)^*.\end{equation}
Every $\Q$-consistent risk measure satisfies
\[\rho(X)\ge\rho(\E_\Q[X])=\E_\Q[X]=\E_\P\big[\tfrac{{\rm d}\Q}{{\rm d}\P}X],\quad X\in L^\infty,\]
which means that $\rho\ge \E_\Q[\cdot]$ and $\rho^*\big(\tfrac{{\rm d}\Q}{{\rm d}\P}\big)=0$.

\begin{definition}[cf.\ \cite{Starshape}]
A normalised monetary risk measure $\rho\colon L^\infty\to\R$ is {\em star shaped} if the acceptance set $\CA_\rho$ is a star-shaped set, i.e., $sY\in\CA_\rho$ holds for all $Y\in\CA_\rho$ and all $s\in[0,1]$.
\end{definition}

Star-shaped risk measures have recently been studied in detail in \cite{Starshape}, and in our definition we implicitly invoke their characterisation provided in \cite[Proposition 2]{Starshape}. Each normalised convex monetary risk measure is star shaped; like consistency, star-shapedness is a weaker property than convexity.  
We also remark that the class of star-shaped consistent risk measures is discussed in \cite[Theorem 11]{Starshape}. While they will mostly play a technical role in our study, further background on them is provided in Lemma~\ref{lem:star}.

\subsection{Admissibility of consistent risk measures}

We now isolate which $\Q$-consistent risk measures $\rho$, $\Q\ll\P$, are admissible for our main results. As preparation, suppose that  $\phi\in(L^\infty)^*$. 
Like in the convex case one verifies on the basis of  cash-additivity and monotonicity  that $\rho^*(\phi)<\infty$ only if $\phi$ is positive --- $\phi|_{L^\infty_+}\ge 0$ --- and $\phi(1)=1$. As usual, we identify $L^1$ with a (proper) subspace of $(L^\infty)^*$ using the pairing $L^1\times L^\infty\ni(Z,X)\mapsto\E_\P[ZX].$
Note that the reference probability under which we integrate is $\P$, even though we may consider a $\Q$-consistent risk measure, $\Q\neq\P$. 
We conclude that $Z\in\dom(\rho^*)\cap L^1$ only if $Z\in L^1_+$ and $\E_\P[Z]=1$; i.e., these elements are probability densities. 

\begin{definition}\label{def:compatible}
A $\Q$-consistent risk measure $\rho\colon L^\infty\to\R$ is \emph{admissible} if there is $Z^*\in L^1$ with the following two properties:
\begin{enumerate}[(a)]
\item $\rho^*(Z^*)<\infty$. 
\item If $U\in\CA^\infty_\rho$ satisfies $\E_\P[Z^*U]=0$, then $U=0$.
\end{enumerate} 
$\mbf C(\rho)$ denotes the set of all \emph{compatible} $Z^*\in L^1$, i.e., $Z^*$ has properties (a)--(b) above. 
\end{definition}

\begin{lemma}\label{lem:is equiv}
Suppose $\Q\ll\P$ and that $\rho\colon L^\infty\to\R$ is an admissible $\Q$-consistent risk measure. Then the following assertions hold:
\begin{itemize}
    \item[(1)]$\Q\approx\P$.
    \item[(2)]Each $Z^*\in\mbf C(\rho)$ satisfies $\P(Z^*>0)=1$.
    \item[(3)]For all $Z\in \dom(\rho^*)\cap L^1$, $Q\in\mbf C(\rho)$, and $\lambda\in (0,1]$,
\begin{equation}\label{eq:is compatible}\lambda Q+(1-\lambda)Z\in\mbf C(\rho).\end{equation}
\end{itemize}
\end{lemma}

Definition~\ref{def:compatible} is best illustrated in the special case where $\rho$ is convex.
By \cite[Proposition 1.1]{continuity}, $\dom(\rho^*)\cap L^1$ is rich under these assumptions. 
Moreover, convexity entails that the asymptotic cone $\CA_\rho^\infty$ collects acceptable net losses $U$ of particular quality; arbitrary quantities $tU$, $t>0$, thereof are still acceptable. 
If $\E_\P[ZU]=0$ for a strictly positive probability density $Z\in L^1_+$ and $U$ is not constant, then it must take negative and positive values, thus triggering net gains {\em and} net losses. 
If $\rho$ is sufficiently risk averse and gains can never fully compensate large losses obtaining in different states, this should disqualify $U$ from being acceptable in arbitrary volumes. 

The notions of admissibility and compatibility are of central importance for the second step in the quest for optimal allocations outlined in the introduction: extracting a convergent optimising sequence. They thus appear prominently in Theorems~\ref{thm:risksh_mon} and \ref{thm:application2} below. 
The next proposition presents a more concrete description of admissible risk measures.

\begin{proposition}\label{compatible nonempty}
Let $\Q\approx\P$ be a probability measure and $\rho\colon L^\infty\to\R$ a $\Q$-consistent risk measure. Then the following statements are equivalent:
\begin{itemize}
\item[(1)]$\rho$ is admissible.
\item[(2)]$\tfrac{{\rm d}\Q}{{\rm d}\P}\in\mbf C(\rho)$.
\item[(3)]$\dom(\rho^*)$ contains at least two elements. 
\item[(4)]$\dom(\rho^*)\cap L^1$ contains at least two elements.
\end{itemize}
Moreover, they all imply
\begin{itemize}
    \item[(5)]There is no constant $\beta>0$ such that $\rho\le\E_\Q[\cdot]+\beta$.
\end{itemize}
\end{proposition}

Note that point (5) in Proposition~\ref{compatible nonempty} means that an admissible consistent risk measure is necessarily more conservative than determining the expected loss under the reference measure $\Q$ and adding a safety margin. 
It therefore provides a simple check to {\em disprove} admissibility. However, it is only a sufficient, not a necessary condition; cf.\ Example~\ref{ex:onlyone}.

An even simpler positive condition guaranteeing admissibility is available if $\rho$ is additionally star shaped.  
Admissibility then boils down to the simple demand that $\rho$ does not agree with the expectation under $\Q$. 
The assumptions of Theorems~\ref{thm:risksh_mon} and \ref{thm:application2} are therefore particularly mild for star-shaped consistent risk measures.

\begin{proposition}\label{uncountably many}
Let $\Q\approx\P$ be a probability measure and $\rho\colon L^\infty\to\R$ be a $\Q$-consistent and star-shaped  risk measure.
Then the following are equivalent: 
\begin{itemize}
\item[(1)]$\rho\neq\E_\Q[\cdot]$.
\item[(2)]$\rho$ is admissible.
\end{itemize}
\end{proposition}

Note that Example~\ref{ex:onlyone} shows that the preceding Proposition~\ref{uncountably many} fails without the assumption of star-shapedness. Further illustrative examples are provided in Example~\ref{ex:many} below.

\begin{remark}\
\begin{enumerate}[(1)]
\item Suppose $\rho$ is admissible. Then \eqref{eq:is compatible} together with points (2) and (4) of Proposition~\ref{compatible nonempty} yields that $\mbf C(\rho)$ contains uncountably many different elements and is dense in $\dom(\rho^*)\cap L^1$. 
\item From Propositions~\ref{compatible nonempty} and~\ref{uncountably many} one can conclude that a $\Q$-consistent risk measure $\rho$ is admissible if and only if its {\em biconjugate} $\rho^{**}\colon L^\infty\to\R$ defined by 
\[\rho^{**}(X)=\sup\{\phi(X)-\rho^*(\phi)\mid \phi\in\dom(\rho^*)\}\]
satisfies $\rho^{**}\neq\E_\Q[\cdot]$. This is due to $\rho^{**}$ being the ``convexification" of $\rho$, i.e., the largest convex $\Q$-consistent risk measure dominated by $\rho$ (cf.\ proof of \cite[Lemma F.3]{Collapse}). 
\item The spirit of Definition~\ref{def:compatible} and, more specifically, equivalence between (1) and (2) established in Proposition~\ref{uncountably many} establishes a link to the realm of ``collapse to the mean" results. The latter subsumes incompatibility results between law invariance of a functional and the existence of ``directions of linearity". The asymptotic cone of the acceptance set of a consistent risk measure can be understood to collect such directions. 
In the theory of risk measures, ``collapse to the mean results" go back at least to \cite{Emanuela}. We refer to the more detailed discussion in \cite{Collapse}. 
\end{enumerate}
\end{remark}

\section{Reference probabilities}\label{sec:heterogeneous}

Throughout the remainder of the paper, $\rho_i\colon L^\infty\to\R$ denotes a $\Q_i$-consistent risk measure used by agent $i\in[n]$, $n\ge 2$, for the assessment of net losses.
However, the reference measure $\Q_i$ may be specific to agent $i$.
The next crucial structural assumption therefore concerns the vector $(\Q_1,\ldots,\Q_n)$ of reference probability measures on $(\Omega,\CF)$ which are all equivalent to $\P$. 

\begin{assumption}\label{assum1} 
Each density $\tfrac{{\rm d}\Q_i}{{\rm d}\P}$, $i\in[n]$, has a version which is a simple function. 
\end{assumption}
Assumption~\ref{assum1} is equivalent to the existence of a partition $\pi\in\Pi$ such that 
\begin{equation}\label{used once}\Q_i(E)=\sum_{B\in\pi}\Q_i(B)\P^B(E),\quad E\in \CF.\end{equation}
Consider the finite set 
$\mathcal P:=\{\P^B\mid B\in\pi\}$
of mutually singular probability measures. 
If two random variables $X,Y\in L^\infty$ agree in distribution under each $\P^B\in\CP$, \eqref{used once} implies that $X\sim_{\Q_i} Y$, $i\in[n]$. 
Hence, $\rho_1,\ldots,\rho_n$ are {\em $\mathcal P$-based risk measures}, a notion recently introduced in \cite{WangZiegel}. Whenever two random variables $X,Y\in L^\infty$ satisfy $X\overset d =_{\P^B}Y$ for all $\P^B\in\mathcal P$, then $\rho_i(X)=\rho_i(Y)$, $i\in[n]$. The present case of mutually singular measures in $\mathcal P$ enjoys particular prominence in \cite{WangZiegel}. The main results in Section~\ref{sec:main} will make heavy use of Assumption~\ref{assum1} to adapt the procedure of comonotonic improvement to the present heterogeneous setting; cf.\ Appendix~\ref{sec:local}, in particular Remark~\ref{Hans}.
We therefore discuss Assumption~\ref{assum1} in more detail from four conceivable angles. 

First, it seems immediate to ask how Assumption~\ref{assum1} is  reflected by the $\Q_i$-consistent risk measures $\rho_i$ used by the agents in the risk sharing problem. 
A complete answer is provided by Theorem~\ref{thm:Pbased}, the main result of this section. 
It translates Assumption~\ref{assum1} as a relaxed notion of dilatation monotonicity common to all $\rho_i$. For more background on dilatation monotonicity, we refer to \cite{Cherny,Rahsepar}, and the references therein.

\begin{definition}
A function $\ph\colon L^\infty\to\R$ is {\em dilatation monotone above} a sub-$\sigma$-algebra $\CH\subset\CF$ if, for all $X\in L^\infty$ and all sub-$\sigma$-algebras $\CH\subset\CG\subset\CF$, 
\begin{equation}\label{eq:above}\ph(\E_\P[X|\CG])\le \ph(X).\end{equation}
\end{definition}
Note that the conditional expectation on the left-hand side of \eqref{eq:above} is computed under $\P$. If $\ph\colon L^\infty\to\R$ is dilatation monotone with respect to $\P$, it is therefore also dilatation monotone above each sub-$\sigma$-algebra $\CH\subset\CF$.

\begin{theorem}\label{thm:Pbased}
Suppose that $\rho_i\colon L^\infty\to\R$ are $\Q_i$-consistent risk measures, $\Q_i\approx\P$, $i\in[n]$.
Then the following are equivalent: 
\begin{itemize}
\item[(1)]The probability measures $(\Q_1,\ldots,\Q_n)$ can chosen to satisfy Assumption~\ref{assum1}.
\item[(2)]Each $\rho_i$ is dilatation monotone above a common finite $\sigma$-algebra $\CH\subset\CF$.
\end{itemize}
\end{theorem}

If $\CG$ is coarser than $\CF$, the conditional expectation $\E_\P[X|\CG]$ displays less variability than the initial random variable $X$. In this sense, $\CG$ can be seen as a gain of information while $\E_\P[X|\CG]$ is usually interpreted as the ``best approximation" of $X$ {\em under the probability measure $\P$} using the information provided by $\CG$. 
Dilatation monotonicity of a risk measure now rewards decreased variability by not increasing the measured risk. 
Item (2) in Theorem~\ref{thm:Pbased} retains this intuition \textit{provided} that the information is sufficient to decide which one of a set $\CH$ of reference events occurs.
In that case, each $\rho_i$ rewards replacing $X$ by its best approximation $\E_\P[X|\CG]$ under the universally shared probability model $\P$.
Otherwise, the information is deemed too coarse, and agents withdraw to their potentially heterogeneous models $\Q_i$. 

Second, each event $B$ in the finite measurable partition $\pi$ from Assumption~\ref{assum1} can be understood as the occurrence of an exogenous shock or a test event used in a backtesting procedure.  
\cite{Asimitetal}, for instance, speak about ``exogenous environments".
While agents disagree about the likelihood of those shocks, their respective relevance or conditional distributional implications are consensus.
For instance, \cite{Marshall} studies optimal insurance contracts under belief heterogeneity. 
A random loss is modelled by a random variable $X\ge 0$, the decision maker expresses probabilistic beliefs with a probability measure $\P$, the insurer with a probability measure $\Q$. One of the case studies in that paper, cf.\ \cite[Section 2]{Marshall}, assumes in our terminology that
\begin{align}\label{eq:Marshall1}\Q(X=0)&>\P(X=0)\quad\quad\text{and}\\
\P^{\{X>0\}}\circ X^{-1}&=\Q^{\{X>0\}}\circ X^{-1}\nonumber\end{align}
\eqref{eq:Marshall1} means that the decision maker is more optimistic about the absence of losses than the insurer. 
Comparing this to \eqref{used once}, one sees that the events $\{X=0\}$ and $\{X>0\}$ could play the role of shocks whose occurrence decision maker and insurer have potentially diverging opinions about, but whose consequences for the conditional distribution of $X$ are uniform. 
In a backtesting context,  \cite{Cambou} impose Assumption~\ref{assum1} {\em verbatim} in their study of {\em scenario aggregation}. The latter problem is faced by a financial company validating (or rejecting) their internal probabilistic model on the basis of evaluating selected adverse test events sufficiently conservatively; cf.\ \cite[Section 5]{Cambou}.
Last, but not least, \cite{WangZiegel} motivate scenario-based risk measures in the same vein.

Third, Assumption~\ref{assum1} is reminiscent of the famous {\em Anscombe-Aumann model} for the study of decisions under uncertainty; cf., e.g., \cite[Chapter 15]{Gilboa}. Without going into too much detail, this model's appeal is the clear separation between {\em ambiguity} about the occurrence of one in a finite set of baseline states, and {\em risk}, i.e., the random payoff drawn according to a state-dependent distribution (called ``lottery"). While this approach canonically models {\em acts} as stochastic kernels, the preceding setting can be interpreted as its translation into random variables.
The events in partition $\pi$ in Assumption~\ref{assum1} correspond to the states. Each random variable $X\in L^\infty$ induces a finite set of lotteries $\{\P^B\circ X^{-1}\mid B\in\pi\}$ on $\R$, depending on the particular state that obtains. These lotteries are ``objective" and recognised by all agents. The likelihoods $\Q_i(B)$ of the occurrence of states $B\in\pi$ are subjective and subject to potential disagreement. 
To our knowledge, this setting is rarely used for risk measures, exceptions being \cite{Drapeau} and \cite[Section 6]{FoeWeb}.

Fourth and last, the assumption can be phrased in the language of statistics. Recall that a sub-$\sigma$-algebra $\CH\subset\CF$ is {\em sufficient} for $\{\Q_1,\ldots,\Q_n\}$ if all bounded random variables $f\colon \Omega\to\R$ admit a common version of the conditional expectation $\E_{\Q_i}[f|\CH]$, $i\in[n]$. On an interpretational level, even the reduced information encoded in the small $\sigma$-algebra $\CH$ is sufficient to construct statistical decision procedures being able to distinguish between the elements in $\{\Q_1,\ldots,\Q_n\}$.

\begin{proposition}\label{prop:sufficient}
A vector $(\Q_1,\ldots,\Q_n)$ of equivalent probability measures on $(\Omega,\CF)$ satisfies Assumption~\ref{assum1} under $\P=\Q_1$ if and only if there is a finite sub-$\sigma$-algebra $\CH\subset\CF$ sufficient for $\{\Q_1,\ldots,\Q_n\}$.
\end{proposition}
For a more thorough discussion of the relation between statistics, decision making, and risk analysis, we refer to \cite{Statistics}.

In view of Proposition~\ref{prop:sufficient} it should be clear that under Assumption~\ref{assum1} $\P$ only plays the role of a weak gauge, determining with its null sets the equivalence class of nonatomic probability measures within which all $\Q_i$, $i\in[n]$, are located. 
In fact,  $\P$ can always be chosen among the $\Q_i$'s.
For our analysis it therefore serves as a somewhat arbitrary point of reference. 
However, as illustrated by Lemma~\ref{lem:is equiv} or Example~\ref{ex:big}(1) below, the assumption of equivalence among the $\Q_i$'s cannot be dropped. This should not surprise, as in a risk sharing scheme an agent could otherwise take on arbitrarily bad outcomes on a null set from their own perspective that is a relevant ground for improvement for other agents involved. The existence of such splitting procedures should rightly preclude the existence of optimal risk sharing schemes.

\section{The main results}\label{sec:main}

\subsection{Risk sharing with consistent risk measures}\label{sec:mon}

All results of this section are proved in Appendix~\ref{proofs of sec 3}. For $X\in L^\infty$, we denote by
\[\mathbb A_X:=\{\mbf X\in (L^\infty)^n\mid X_1+\ldots+X_n=X\}\]
the set of all {\em allocations} of $X$. 
Our problem of interest concerns $\Q_i$-consistent risk measures $\rho_i\colon\Linfty\to\R$, where $\Q_i\approx\P$ is some probability measure, $i\in[n]$.
For $X\in\Linfty$, we aim to solve the problem 
\begin{equation}\begin{array}{l}\label{problem1}\sum_{i=1}^n\rho_i(X_i)\longrightarrow\min\\
\text{subject to }\mbf X\in\mathbb A_X.\end{array}\end{equation}
The associated \emph{infimal convolution} 
$\rho:=\Box_{i\in[n]}\rho_i$
gives precisely the optimal value of \eqref{problem1}. The functional $\rho$ is known to be cash-additive and monotone, i.e., $\rho$ is a monetary risk measure if $\dom(\rho)\neq\varnothing$. 
An allocation $\mbf X\in\mathbb A_X$, $X\in L^\infty$, is {\em optimal} if $\mbf X$ is an optimiser in problem \eqref{problem1}, i.e.,
\[\rho(X)=\left(\Box_{i\in[n]}\rho_i\right)(X)=\sum_{i=1}^n\rho_i(X_i).\]
If an optimal allocation of $X$ exists, we say that $\rho$ is {\em exact} at $X$. 

\begin{theorem}\label{thm:risksh_mon}
Suppose:
\begin{itemize}
\item[(i)] A vector $(\Q_1,\ldots,\Q_n)$ of $\P$-equivalent probability measures satisfies Assumption~\ref{assum1}.
\item[(ii)]For each $i\in[n]$, $\rho_i\colon L^\infty\to\R$ is a $\Q_i$-consistent risk measure. 
\item[(iii)]The risk measures $\rho_1,\ldots,\rho_{n-1}$ are admissible, and for all $i\in[n-1]$ there is $Z_i\in\mbf C(\rho_i)$ such that, for all $k\in[n]$, $Z_i\in\dom(\rho_k^*)$. 
\end{itemize}
Then, for each $X\in L^\infty$, there is an optimal allocation $\mbf X\in\mathbb A_X$.
\end{theorem}

\subsubsection{Discussion of assumption (iii)}\label{sec:discussion}

Assumptions (i) and (ii) of Theorem~\ref{thm:risksh_mon} have been discussed in detail in Sections~\ref{sec:consistent}--\ref{sec:heterogeneous}. We therefore turn directly to assumption (iii). 
Assume for the moment that {\em all} risk measures $\rho_1,\ldots,\rho_n$ are admissible. Then (iii) is implied by the simpler condition 
\begin{equation}\label{eq:special}\forall\,i,k\in[n]:\quad\rho^*_k(\tfrac{{\rm d}\Q_i}{{\rm d}\P})<\infty.\end{equation}
This requirement has a clear interpretation if all $\rho_i$'s are classical convex risk measures, in which case the probability densities in $\dom(\rho_i^*)\cap L^1$ are ``plausible probabilistic models that are taken more or less seriously" (\cite[p.\ 308]{FoeWeb}) by agent $i$ depending on the size of $\rho^*_i$. Each reference model $\frac{{\rm d}\Q_i}{{\rm d}\P}$ plays a fundamental role for agent $i$, which is underscored by the fact that $\rho_i^*\left(\tfrac{{\rm d}\Q_i}{{\rm d}\P}\right)=0$, i.e., $\tfrac{{\rm d}\Q_i}{{\rm d}\P}$ enjoys maximal plausibility from the point of view of $i$. 
\eqref{eq:special} then means that $\frac{{\rm d}\Q_i}{{\rm d}\P}$ is also deemed somewhat plausible by the other agents. 
We can adopt this thinking in the more general consistent case, in which formula~\eqref{eq:conjugate} holds for $\rho_i^*$. 
However, Example~\ref{ex:add1} demonstrates that two admissible risk measures $\rho_1$ and $\rho_2$ can satisfy all assumptions of Theorem~\ref{thm:risksh_mon}, but fail \eqref{eq:special}. 

While~\eqref{eq:special} is too restrictive, its perspective can inform the interpretation of assumption (iii) itself. Given $i\in[n-1]$ and the admissible risk measure $\rho_i$,
its compatible elements $\mbf C(\rho_i)$ are dense in $\dom(\rho_i^*)\cap L^1$, but more plausible in the sense of being less extreme than models outside $\mbf C(\rho_i)$. Assumption (iii) means that within this set, we find a density $Z_i$ that enjoys a certain degree of confidence by all agents involved. 

Another noteworthy aspect of assumption (iii) is that only $\rho_1,\ldots,\rho_{n-1}$ are assumed to be admissible. This permits to include, for instance, expected loss assessments that price risk linearly, i.e., the $\Q_n$-consistent, but not admissible risk measure $\rho_n=\E_{\Q_n}[\cdot]$. 
If $\rho_n$ is not admissible, Proposition~\ref{compatible nonempty}(3) simplifies  assertion (iii)  substantially to the requirement
\begin{equation}\label{eq:intersection}\tfrac{{\rm d}\Q_n}{{\rm d}\P}\in\bigcap_{i=1}^{n-1}\mbf C(\rho_i).\end{equation}
Moreover, condition \eqref{eq:special} is precluded unless we are in a homogeneous situation.
We can nevertheless also solve the optimal allocation problem if {\em more than one} agent measures risk with a non-admissible risk measure. The crucial condition is that all non-admissible agents share the same reference probability measure. 
Assume that $\Q_1,\ldots,\Q_n$ are the reference measures described by Assumption~\ref{assum1}. Let $K\in\N$ and suppose $\tau_1,\ldots,\tau_K$ are $\Q_n$-consistent risk measures which are not admissible, while $\rho_1,\ldots,\rho_{n-1}$ are admissible consistent risk measures.
Set $\w\rho:=\Box_{j=1}^K\tau_j$ and note that by Proposition~\ref{compatible nonempty}(3), \[\dom(\w\rho^*)=\dom\Big(\sum_{j=1}^K\tau_j^*\Big)=\bigcap_{j=1}^K\dom(\tau_j^*)=\{\tfrac{{\rm d}\Q_n}{{\rm d}\P}\}.\]
Assumptions (i)--(iii) of Theorem~\ref{thm:risksh_mon} are then met by the agents $(\rho_1,\ldots,\rho_K,\w\rho)$ if and only if \eqref{eq:intersection} holds, and we may allocate $X\in L^\infty$ optimally among them by Theorem~\ref{thm:risksh_mon}.
In a second step, let $Y_n$ be the share attributed to the fictitious last agent measuring risk with $\w\rho$.
\cite[Theorem 4.1]{Consistent} proves the existence of optimal allocations for $\Q_n$-consistent risk measures, i.e., for the case of homogeneous beliefs, and also admits ``linear agents" mentioned above. 
Thus we can choose $Z_1,\ldots,Z_K\in L^\infty$ with the properties
\[Y_n=\sum_{i=1}^{K}Z_i\quad\text{and}\quad\sum_{i=1}^K\tau_i(Z_i)=\w\rho(Y_n).\]
Setting $X_i$ to be the share initially attributed to agent $i\in[n-1]$, it is straightforward to show that $(X_1,\ldots,X_{n-1},Z_1,\ldots,Z_K)$ is an optimal allocation of $X$ among $\rho_1,\ldots,\rho_{n-1},\tau_1,\ldots,\tau_K$. 

\subsubsection{Necessity of assumptions}

The homogeneous case also shows that law invariance alone or combined with star-shapedness is too weak to guarantee the existence of optimal allocations.\footnote{~I am indebted to Ruodu Wang for pointing out the corresponding result in~\cite{Many}.}
As an illustration, consider constants $\alpha,\beta>0$ with $\alpha+\beta<1$. Let $\rho_i\colon L^\infty\to\R$, $i=1,2$, be defined by 
\[\begin{array}{l}\rho_1(X):=\inf\{x\in\R\mid \P(X\le x)>1-\alpha\},\\
\rho_2(X):=\tfrac 1 {\beta}\int_{1-\beta}^1 q_X(s)\,{\rm d}s.\end{array}\]
Here $q_\bullet$ denotes a version of the quantile function under $\P$. Both functionals are law invariant with respect to the reference measure $\P$, cash-additive, normalised, positively homogeneous and thus star-shaped. However, $\rho_2$ is a consistent risk measure, while $\rho_1$ is not consistent. 
At last, no $X\in L^\infty$ with continuous distribution can be allocated Pareto optimally between $\rho_1$ and $\rho_2$ (\cite[Proposition 1]{Many}). 

Example~\ref{ex:big} continues the illustration of the necessity of the assumptions of Theorem~\ref{thm:risksh_mon} in case $n=2$.

\subsubsection{Shape of optimal allocations}\label{sec:shape}

One of the key findings in many monetary risk sharing situations under a single homogeneous reference measures is that {\em comonotone} optimal risk allocations can be found.  
More precisely, let $\mf C$ denote the set of all {\em comonotone functions} functions $\mbf f:\R\to\R^n$ satisfying that each coordinate $f_i$ is non-decreasing and that $\sum_{i=1}^nf_i=id_\R$. 
Consequently, each coordinate $f_i$ of $\mbf f\in\mf C$ is 1-Lipschitz continuous. 
For $\mbf f\in\mf C$, we abbreviate by $\widetilde{\mbf f}$ the normalised function $\mbf f-\mbf f(0)$.
Then comonotone allocations $\mbf f(X)$ are obtained by applying a suitable $\mbf f\in\mf C$ to the aggregated quantity $X$.

The proofs in Appendix~\ref{proofs of sec 3} demonstrate that optimal allocations will usually {\em not} be comonotone under heterogeneous reference measures, marking a substantial difference between the heterogeneous and homogeneous  case.
As an illustration, consider an arbitrary $\Q\approx\P$, positive constants $\beta_1,\beta_2>0$, and the convex risk measures
\[\begin{array}{l}\rho_1(X):=\tfrac{1}{\beta_1}\log(\E_\P[e^{\beta_1X}]),\\
\rho_2(X):=\tfrac1{\beta_2}\log(\E_\Q[e^{\beta_2X}]),\end{array}\quad\quad X\in\Linfty.\]
Then $\rho_1\Box\rho_2$ is exact at each $X\in\Linfty$, and the unique optimal risk allocation $(X_1,X_2)$ is given by
\[\begin{array}{l}X_1=\tfrac{\beta_2}{\beta_1+\beta_2}X+\tfrac{1}{\beta_1+\beta_2}\log(\tfrac{{\rm d}\Q}{{\rm d}\P}),\\
X_2=\tfrac{\beta_1}{\beta_1+\beta_2}X-\tfrac{1}{\beta_1+\beta_2}\log(\tfrac{{\rm d}\Q}{{\rm d}\P}).\end{array}\]
The dependence on the density $\tfrac{{\rm d}\Q}{{\rm d}\P}$ precludes comonotonicity \textit{unless} $\tfrac{{\rm d}\Q}{{\rm d}\P}\equiv 1$, which is the case if and only if $\Q=\P$.

In order to justify the desirability of comonotone allocations, let us decompose each risk portion $f_i(X)$ of a comonotone allocation $\mbf f(X)$ into two parts: $f_i(0)$, a deterministic cost imposed on or capital injected into the position of agent $i$; and the remainder $\widetilde{f}_i(X)$. In actuarial terminology, $\widetilde f_i$ is an \textit{indemnity function}. $\widetilde f_i(X)$ reflects minimal rationality and fairness considerations of the agents involved in the risk sharing scheme. 
If $X>0$ (i.e., a loss is produced), $0\le \widetilde f_i(X)\le X$; if $X<0$ (i.e., a gain is obtained), then $X\le \widetilde f_i(X)\le 0$. 
The allocated net losses thus do not exceed the total net loss, and no agent can symmetrically make a profit if the system at large suffers a loss. 
Moreover, in the tradition of \cite{Huberman}, the nondecreasing nature of $\widetilde f_i$ is interpreted as the absence of {\em ex post} moral hazard potentially incentivising agents to misreport their losses.\footnote{~As \cite[p.\ 423]{Huberman} write, 
``[o]nce damage has occurred, they encourage the insured both to create further damage and to inflate the size of the claim. For example, in automobile insurance, once an automobile is damaged, if further damage cannot be distinguished as extra damage, the insurance company will compensate for the combined damage. Moreover, disappearing deductibles reinforce incentives to do things like present inflated repair estimates and demand more expensive health care."}
By definition, the deterministic cash transfers $(f_i(0))_{i\in[n]}\in\R^n$ satisfy $\sum_{i=1}^nf_i(0)=0$. Given that individual risk measures are cash-additive, deterministic capital transfers among agents---which could be perceived as unfair---can be eliminated without losing optimality, leading to the new optimal allocation $\widetilde{\mbf f}(X)$ of $X$. 
This ``rebalancing of cash" can alternatively be justified by solving a second optimisation problem and selecting a optimal comonotone function $\mbf g\in\mf C$ whose cash transfers $\mbf g(0)$ are  closest to a uniform distribution among agents; i.e., one additionally minimises the function  
\[\Xi\colon\begin{array}{l}\mf C\to\R,\\
\mbf f\mapsto\min\{\sum_{i=1}^n|f_i(0)-x|^2\mid x\in\R\},\end{array}\]
among all $\mbf f\in\mf C$ defining an optimal allocation. 

In the heterogeneous case and under Assumption~\ref{assum1}, the optimal allocations whose existence is verified in Theorem~\ref{thm:risksh_mon} are only {\em locally comonotone} (Definition~\ref{def:locally comonotone}); that is, if $\pi\in\Pi$ is as in \eqref{used once}, one finds an optimal  family $(\mbf f^B)_{B\in\pi}\subset\mf C$ such that 
$X_i:=\sum_{B\in\pi}f_i^B(X)\ind_B$, $i\in[n]$,
defines an optimal allocation. Two observations ensue: The potential loss of comonotonicity is a consequence of the heterogeneity of reference measures, not (primarily) of properties of the risk measures. Second, locally comonotone allocations split into a $\pi$-dependent cash transfer reflecting the heterogeneity structure, $\sum_{B\in\pi}f_i(0)\ind_B$, and $B$-dependent indemnity schedules applied to the aggregated loss $X$, $\sum_{B\in\pi}\widetilde f_i(X)\ind_B$. 

Note that the risk measures $\rho_i$, $i\in[n]$, cannot behave additively on span$\{\ind_B\mid B\in\pi\}$ by \cite[Theorem 5.7]{Collapse}. Hence, rebalancing of cash and eliminating cash transfers among agents in the way described above becomes impossible. 
This important distinctive feature of the heterogeneous case has been emphasised multiple times in this manuscript already. 
As a workaround, one can take the alternative approach to rebalancing cash, select a quality criterion $\Xi$ measuring ``unfairness" and minimising it among all optimal allocations. 
\begin{corollary}\label{cor:problem2}
In the situation of Theorem~\ref{thm:risksh_mon} let $\pi\in\Pi$ be as in \eqref{used once}. Let $\Xi\colon\mf C\to\R$ be bounded below and continuous with respect to sequential pointwise convergence on $\mf C$. 
Consider the optimisation problem 
\begin{equation}\begin{array}{l}\label{problem2}\Xi\left((\mbf f^B)_{B\in\pi}\right)\longrightarrow\min\\
\text{subject to }(\mbf f^B)_{B\in \pi}\in\mf C^\pi,~\sum_{i=1}^n\rho_i\left(\sum_{B\in\pi}f_i^B(X)\ind_B\right)=(\Box_{i=1}^n\rho_i)(X).\end{array}\end{equation}
Then \eqref{problem2} has a solution. 
\end{corollary}

\subsubsection{Infimal convolution preserves $\CP$-basedness}

Let us once again consider the situation of a homogeneous reference measure $\P$ in which comonotone optimal allocations can be found. 
As an important consequence, one can show that the infimal convolution preserves $\P$-law invariance, i.e., the optimal value in \eqref{problem1} depends only on the distribution of the aggregated loss $X$ under $\P$. 
However, Theorem~\ref{thm:risksh_mon} also applies to heterogeneous reference measures. Moreover, the lack of comonotonicity of optimal allocations outlined in Section~\ref{sec:shape} shows that one cannot expect the infimal convolution to be law invariant with respect to some reference measure. 
Instead, Corollary~\ref{cor:Pbased} shows that the infimal convolution operation preserves {\em $\mathcal P$-basedness} of the individual risk measures discussed in Section~\ref{sec:heterogeneous}. This is an important addendum to \cite{WangZiegel}. 
\begin{corollary}\label{cor:Pbased} 
In the situation of Theorem~\ref{thm:risksh_mon} let $\pi\in\Pi$ be a finite measurable partition as in \eqref{used once}. Consider the set
\[\CP:=\{\P^B\mid B\in\pi\}\]
of conditional probability measures. Then the infimal convolution $\rho=\Box_{i=1}^n\rho_i$ is $\CP$-based: If for $X,Y\in L^\infty$ and all $B\in\pi$ we have $X\sim_{\P^B}Y$, then $\rho(X)=\rho(Y)$. 
\end{corollary}

\subsection{Risk sharing with capital requirements}

Proofs and auxiliary results accompanying this subsection---which presents the second main result of the paper---can be found in Appendix~\ref{last appendix}. 
Here we focus on risk sharing with capital requirements based on multidimensional security spaces---sometimes also called {\em multi-asset risk measures}. 
The formulation of this problem---which is thoroughly studied in \cite{Liebrich} under the assumption of convexity and has similarities to the C\&R transfers suggested in \cite{FilKup}---requires \textit{risk measurement regimes},
a notion we adopt from the aforementioned paper to the present setting. 

For a vector $(\Q_1,\ldots,\Q_n)$ of probability measures equivalent to $\P$, $\rho_i\colon L^\infty\to\R$ are $\Q_i$-consistent risk measures,
$i\in[n]$. 
\emph{Security spaces} are finite-dimensional subspaces $\Scal_i\subset\Linfty$ which contain a nontrivial $S_i\ge 0$, $i\in[n]$. Elements $S\in\mathcal S_i$ are called  \emph{securities}.
$\price_i:\Scal_i\to\R$ are linear and positive \textit{pricing functionals}, mapping at least one positive security to a positive price. 
If these assumptions are satisfied and, for all $X\in L^\infty$,
\[\sup\{\price_i(S)\mid S\in\mathcal S_i,\ X+S\in\CA_{\rho_i}\}<\infty,\]
$\mathcal R_i:=(\CA_{\rho_i},\Scal_i,\price_i)$ is a \textit{risk measurement regime}, $i\in[n]$, and it is called  {\em admissible} if $\rho_i$ is an admissible $\Q_i$-consistent risk measure. 

These risk measurement regimes induce \textit{capital requirements} defined---in the spirit of \cite{Baes,FKM2015,FritScand}---as 
\[\eta_i(X):=\inf\{\price_i(S)\mid S\in\mathcal S_i,\,X-S\in\CA_{\rho_i}\}.\]
The value $\eta_i(X)$ has an immediate operational interpretation as the infimal amount of capital which needs to be raised and invested in suitable securities $S\in\mathcal S_i$ available in the security market at price $\price_i(S)$ such that the secured net loss profile $X-S$ passes the capital adequacy test posed by $\CA_{\rho_i}$. Note also that the risk measures $\rho_i$ merely play the role of gauges determinining acceptability. Their precise values do not matter for the computation of the $\eta_i$'s. 
In contrast to monetary risk measures, the capital requirements $\eta_i$ may neither only attain finite values nor be cash-additive. In Theorem~\ref{thm:application2} below, we shall impose the following assumption.

\begin{assumption}\label{assum:RMR}
There is $Z^*\in\bigcap_{i=1}^n\mbf C(\rho_i)$ such that, for all $i\in[n]$ and all $S\in\Scal_i$, 
$$\price_i(S)=\E_\P[Z^*S].$$
\end{assumption}

Concerning the satisfiability of Assumption~\ref{assum:RMR}, suppose that each $\rho_i$ is admissible and $\Q_i$-consistent and that, for all $i,k\in[n]$, $\tfrac{{\rm d}\Q_i}{{\rm d}\P}\in\dom(\rho_k^*)$.
Then $Z^*$ as in Assumption~\ref{assum:RMR} exists; for instance, $Z^*:=\sum_{i=1}^n\delta_i\tfrac{{\rm d}\Q_i}{{\rm d}\P}$ for arbitrary choices of positive convex combination parameters $\delta_i>0$, $i\in[n]$ (Proposition~\ref{compatible nonempty}).

\begin{remark}
Assumption~\ref{assum:RMR} requires further elaboration.
Let us take for a moment the perspective of a market, i.e., random variables represent net gains (or payoffs) instead of net losses. 
$\mathcal S_i$ collects (portfolios of) assets in which agent $i$ can trade. 
Acceptance set $\CA_{\rho_i}$ represents agent $i$'s preferences collecting payoffs they {\em would like to have}; a payoff is desired if it belongs to $-\CA_{\rho_i}$.
Common quality criteria to assess pricing rules are absence of arbitrage opportunities and, more generally, ``absence of good deals". 
A good deal for agent $i$ is a nontrivial desirable payoff that has at most zero cost, i.e., $X\in -\CA_{\rho_i}\cap\{S\in\mathcal S_i\mid \price_i(S)\le 0\}$. 
A key difference to arbitrage opportunities is that these are agent-independent and defined via monotonicity of preferences; good deals, on the other hand, are highly dependent on the agent's preferences. For this reason, the recent working paper \cite{Arduca1} criticises the absence of good deals on economic grounds. Instead, 
the more natural condition is the absence of so-called ``scalable good deals", and a natural counterpart to the classical Fundamental Theorem of Asset Pricing can be established using that condition. 
Absence of scalable good deals means here that nontrivial elements $X\in-\CA_{\rho_i}^\infty\cap\{S\in\mathcal S_i\mid\price_i(S)\le 0\}$ do not exist. 
While Assumption~\ref{assum:RMR} is not the absence of scalable acceptable deals, it is a dual condition that posits the existence of a pricing density preventing  scalable acceptable deals simultaneously for all heterogeneous agents in a market. 
\end{remark}

Under Assumption~\ref{assum:RMR}, the capital requirements $\eta_i$ are typically far from consistent. Even more, they lose law invariance of the acceptance set $\CA_{\rho_i}$, a point already emphasised in the introduction. 
The risk sharing problem in the present subsection is therefore remarkable insofar as it delivers the existence of optimal allocations both without convexity and law invariance assumptions. 
The preceding point is made precise in the following observation:

\begin{proposition}\label{not li}
Suppose $(\CA_\rho,\mathcal S,\price)$ is an admissible risk measurement with associated $\P$-consistent risk measure $\rho$. Moreover, assume that $\price=\phi|_{\mathcal S}$ for some $\phi\in\dom(\rho^*)$.
Then the following are equivalent: 
\begin{itemize}
\item[(1)]$\mathcal S=\R$.
\item[(2)]$\eta$ is $\P$-law invariant.
\end{itemize}
\end{proposition}

\begin{remark}
Proposition~\ref{not li} complements and, where it applies, sharpens substantially the similar result \cite[Proposition 5.8]{Bellini}. Note that the additional condition that $\price=\phi|_{\mathcal S}$ for some $\phi\in\dom(\rho^*)$ is mild and equivalent to the associated capital requirement $\eta$ satisfying $\dom(\eta^*)\neq \varnothing$. For instance, $\phi$ exists if $\eta$ is convex and only takes finite values. 
\end{remark}

Define $\mathcal M:=\sum_{i=1}^n\Scal_i$ and $\pi\colon \mathcal M\to\R$ by $\pi(S)=\E_\P[Z^*S]$, $S\in\mathcal M$. 
The couple $(\mathcal M,\pi)$ forms the \textit{global security market}.
Like in Theorem~\ref{thm:risksh_mon}, we define the infimal convolution $\eta:=\Box_{i=1}^n\eta_i\colon L^\infty\to\R$ by
\[\eta(X)=\inf_{\mbf X\in\mathbb A_X}\eta_i(X_i),\]
to find optimal allocations, 
and ask if $\eta$ takes finite values and is exact at a given $X\in L^\infty$. 
The following theorem generalises \cite[Theorem 5.6]{Liebrich}. 

\begin{theorem}\label{thm:application2}
Suppose: 
\begin{itemize}
\item[(i)] The underlying vector $(\Q_1,\ldots,\Q_n)$ of probability measures satisfies Assumption~\ref{assum1}. 
\item[(ii)] The risk measurement regimes $(\mathcal R_1,\ldots,\mathcal R_n)$ satisfy Assumption~\ref{assum:RMR}.
\item[(iii)] $\mbf U\in\prod_{i=1}^n\Scal_i$ such that $U:=\sum_{i=1}^nU_i$ satisfies $\pi(U)=1$.
\end{itemize}
Then for all $X\in\sum_{i=1}^n\dom(\eta_i)$ there is $\mbf N\in\prod_{i=1}^n\Scal_i$ and $\mbf A\in\prod_{i=1}^n\CA_{\rho_i}$ such that $N:=\sum_{i=1}^nN_i\in\ker(\pi)$,
and $\mbf X\in\mathbb A_X$ defined by 
\[
X_i:=A_i+N_i+\eta(X)U_i,\quad i\in[n],
\]
satisfies 
\begin{equation}\label{eq:shape}\eta(X)=\sum_{i=1}^n\eta_i(X_i).\end{equation}
\end{theorem}

The assumptions of Theorem~\ref{thm:application2} cannot be relaxed, as Example~\ref{ex:add2} demonstrates.

\appendix

\section{Auxiliary results}\label{appendix}

Before we give the proofs accompanying results from Sections~\ref{sec:consistent}--\ref{sec:main}, this appendix collects structural properties of $\Q$-law-invariant risk measures where the probability measure $\Q$ may not agree with the gauge probability measure $\P$. 
While the first set of results in Lemma~\ref{lem:structure} is standard if $\Q=\P$, we shall provide the short proof in the general case for the convenience of the reader.

\begin{lemma}\label{lem:structure}
Let $\Q\approx\P$ be a probability measure and $\rho\colon L^\infty\to\R$ be a $\Q$-consistent risk measure. 
\begin{itemize}
\item[(1)]For all $Q\in L^1$,
\begin{equation*}\rho^*(Q)=\rho^\sharp\big((\tfrac{{\rm d}\Q}{{\rm d}\P})^{-1}Q\big),\end{equation*}
where 
\begin{equation}\label{eq:rhosharp}\rho^\sharp(Z)=\sup_{X\in L^\infty}\E_\Q[ZX]-\rho(X),\quad Z\in L^1_\Q.\end{equation}
Moreover, the function $\rho^\sharp$ is $\Q$-law invariant.
\item[(2)]For all $Q\in L^1\cap L^1_\Q$ and all sub-$\sigma$-algebras $\CG\subset\CF$ such that $\tfrac{{\rm d}\Q}{{\rm d}\P}$ has a $\CG$-measurable version, 
\[\rho^*(\E_\P[Q|\CG])\le \rho^*(Q).\]
\end{itemize}
\end{lemma}
\begin{proof}\
\begin{enumerate}[(1)]
\item Let $Q\in L^1$ and compute
\begin{align*}\rho^*(Q)&=\sup_{X\in L^\infty}\E_\P[QX]-\rho(X)=\sup_{X\in L^\infty}\E_\Q\big[(\tfrac{{\rm d}\Q}{{\rm d}\P})^{-1}QX\big]-\rho(X)=\rho^\sharp\big((\tfrac{{\rm d}\Q}{{\rm d}\P})^{-1}Q\big).
\end{align*}
Law invariance with respect to $\Q$ of the function $\rho^\sharp$ follows, for instance, from \cite[Lemma 3.2]{General}. 
\item Let $\CG\subset\CF$ be a sub-$\sigma$-algebra such that $\tfrac{{\rm d}\Q}{{\rm d}\P}$ has a $\CG$-measurable version. Let $Q\in L^1\cap L^1_\Q$. Using the Bayes rule for conditional expectations, 
\[
\tfrac{{\rm d}\Q}{{\rm d}\P}\E_\Q[Q|\CG]=\E_\P[\tfrac{{\rm d}\Q}{{\rm d}\P}Q|\CG]=\tfrac{{\rm d}\Q}{{\rm d}\P}\E_\P[Q|\CG].
\]
$\Q\approx\P$ implies that $\E_\Q[Q|\CG]=\E_\P[Q|\CG]$. 
The function $\rho^\sharp$ defined in \eqref{eq:rhosharp} is convex, $\Q$-law invariant, proper (i.e., $\dom(\rho^\sharp)\neq\varnothing$), and $\sigma(L^1_\Q,L^\infty_\Q)$-lower semicontinuous by definition. 
Hence, \cite[Theorem 3.6]{General} admits to infer for all $Q\in L^1$ that
\begin{align*}\rho^*(Q)&=\rho^\sharp\big((\tfrac{{\rm d}\Q}{{\rm d}\P})^{-1}Q\big)\ge \rho^\sharp\big(\E_\Q\big[(\tfrac{{\rm d}\Q}{{\rm d}\P})^{-1}Q|\CG\big]\big)=\rho^\sharp\big((\tfrac{{\rm d}\Q}{{\rm d}\P})^{-1}\E_\Q[Q|\CG]\big)\\
&=\rho^*(\E_\Q[Q|\CG])=\rho^*(\E_\P[Q|\CG]).
\end{align*}
\end{enumerate}
\end{proof}

\begin{lemma}\label{lem:direction}
Let $\Q\approx\P$ be a probability measure and $\rho\colon L^\infty\to\R$ be a $\Q$-consistent risk measure. 
\begin{itemize}
\item[(1)]For all $\phi\in\dom(\rho^*)$ and all $U\in\CA_\rho^\infty$, $\phi(U)\le 0$ holds. 
\item[(2)]$\CA_\rho^\infty$ is closed, $\Q$-law-invariant, and closed under taking conditional expectations with respect to $\Q$. 
\end{itemize}
\end{lemma}
\begin{proof}\
\begin{enumerate}[(1)]
\item Let $\phi$ and $U$ be as described. Let $(s_k)_{k\in\N}\subset(0,\infty)$ be a null sequence and $(Y_k)_{k\in\N}\subset\CA_\rho$ such that $\lim_{k\to\infty}s_kY_k=U$. 
Recall that $\rho^*(\phi)=\sup_{Y\in\CA_\rho}\phi(Y)\geq \phi(0)=0$. 
Hence,
\[\phi(U)=\lim_{k\to\infty}s_k\phi(Y_k)\le \rho^*(\phi)\cdot\lim_{k\to\infty}s_k=0.\]
\item It is straightforward to verify that $\CA_\rho^\infty$ is closed. In order to see that $\CA_\rho^\infty$ is closed under taking conditional expectations with respect to $\Q$, fix $U\in\CA_\rho^\infty$ and an arbitrary sub-$\sigma$-algebra $\CG\subset\CF$.
Let  $(s_k)_{k\in\N}\subset(0,\infty)$ and $(Y_k)_{k\in\N}\subset\CA_\rho$ be sequences as in the proof of (1) with $\lim_{k\to\infty}s_kY_k=U$.
By Jensen's inequality and the monotonicity of $\rho$ in second order stochastic dominance relation with respect to $\Q$, $\E_\Q[Y_k|\CG]\in\CA_\rho$ holds for all $k\in\N$. Hence,  
\[\E_\Q[U|\CG]=\lim_{k\to\infty}\E_\Q[s_kY_k|\CG]=\lim_{k\to\infty}s_k\E_\Q[Y_k|\CG]\in\CA_\rho^\infty.\]
The function
$F\colon L^\infty=L^\infty_\Q\to\R_+$ defined by $F(X)=\inf_{U\in\CA_\rho^\infty}\|X-U\|_\infty$
is continuous (because $\CA_\rho^\infty$ is closed) and $\Q$-dilatation monotone. 
Indeed, for every $X\in L^\infty$ and every sub-$\sigma$-algebra $\CG\subset\CF$, 
\begin{align*}
    F\left(\E_\Q[X|\CG]\right)&\le\inf_{U\in\CA_\rho^\infty}\left\|\E_\Q[X|\CG]-\E_\Q[U|\CG]\right\|_\infty=\inf_{U\in\CA_\rho^\infty}\left\|\E_\P[X-U|\CG]\right\|_\infty\le\inf_{U\in\CA_\rho^\infty}\|X-U\|_\infty=F(X).
\end{align*}
Hence, by \cite[Theorem 1.1]{Cherny}, $F$ is $\Q$-law invariant. 
This translates to $\Q$-law-invariance of $\CA_\rho^\infty=F^{-1}(\{0\})$.
\end{enumerate}
\end{proof}

The next lemma reveals that also the set of compatible elements is stable under computing certain conditional expectations.

\begin{lemma}\label{lem:main}
Let $\Q\approx\P$ be a probability measure and $\rho\colon L^\infty\to\R$ be a $\Q$-consistent risk measure. 
Suppose $\CG\subset\CF$ is a sub-$\sigma$-algebra such that $\tfrac{{\rm d}\Q}{{\rm d}\P}$ has a $\CG$-measurable version.
Then, for all $Z\in\mbf C(\rho)$, 
\begin{center}$\E_\P[Z|\CG]\in\mbf C(\rho)$.\end{center}
\end{lemma}
\begin{proof}
As there is nothing to show if $\mbf C(\rho)=\varnothing$, we may assume that $\rho$ is admissible. Let $\CG\subset\CF$ be a sub-$\sigma$-algebra as described. By Lemma~\ref{lem:structure}(2), $\E_\P[Z|\CG]\in\dom(\rho^*)$. Moreover, as $\Q\approx\P$, we may apply the Bayes rule for conditional expectations to compute
\begin{equation}\label{eq:Bayes}\E_\Q[Z|\CG]=\left(\E_\P[\tfrac{{\rm d}\Q}{{\rm d}\P}|\CG]\right)^{-1}\E_\P[\tfrac{{\rm d}\Q}{{\rm d}\P}Z|\CG]=\left(\tfrac{{\rm d}\Q}{{\rm d}\P}\right)^{-1}\cdot\tfrac{{\rm d}\Q}{{\rm d}\P}\E_\P[Z|\CG]=\E_\P[Z|\CG].\end{equation}
Now suppose $U\in\CA_\rho^\infty$ satisfies $\E_\P[\E_\P[Z|\CG]U]=0$. $\E_\P[Z|\CG]\in\mbf C(\rho)$ follows if we can show $U=0$. 
Switching the conditioning and arguing as in \eqref{eq:Bayes} yields 
\[0=\E_\P[\E_\P[Z|\CG]U]=\E_\P[Z\E_\Q[U|\CG]],\]
and $\E_\Q[U|\CG]\in\CA_\rho^\infty$ holds by Lemma~\ref{lem:direction}(2). As $Z\in\mbf C(\rho)$, we infer $\E_\Q[U|\CG]=\E_\Q[U]=0$. As $\tfrac{{\rm d}\Q}{{\rm d}\P}\in\mbf C(\rho)$ by Proposition~\ref{compatible nonempty}(2), $U=0$ has to hold.  
\end{proof}

In the next preparatory result, we consider the operation of computing the star-shaped hull of a normalised monetary risk measure $\rho\colon L^\infty\to\R$, i.e., the functional $\rho_\star\colon L^\infty\to\R$ defined by
\begin{equation}\label{starshaped hull}\rho_\star(X):=\inf\{m\in\R\mid\exists\,s\in(0,1]:~X-m\in s\CA_\rho\}.\end{equation}
\begin{lemma}\label{lem:star}
Suppose $\rho\colon L^\infty\to\R$ is a normalised monetary risk measure with $\dom(\rho^*)\neq\varnothing$.
\begin{itemize}
\item[(1)] $\rho_\star$ is a star-shaped risk measure which satisfies 
\begin{equation}\label{eq:id1}\CA_{\rho_\star}=\tn{cl}\Big(\bigcup_{s\in(0,1]}s\CA_\rho\Big)\quad\text{and}\quad\rho_\star^*=\rho^*.\end{equation}
\item[(2)]If $\rho$ is $\Q$-consistent, then so is $\rho_\star$.
\item[(3)]
The asymptotic cones satisfy
\begin{equation}\label{eq:asymptoticequal}\CA_{\rho_\star}^\infty=\CA_\rho^\infty.\end{equation}
\item[(4)]If $\rho$ is $\Q$-consistent,
\[\mbf C(\rho)=\mbf C(\rho_\star).\]
In particular, $\rho$ is admissible if and only if $\rho_\star$ is admissible.
\end{itemize}
\end{lemma}
\begin{proof}\
\begin{enumerate}[(1)]
\item The condition $\dom(\rho^*)\neq\varnothing$ ensures that $\rho_\star(X)\in\R$ for all $X\in L^\infty$.
For monotonicity, let $X\le Y$ be two random variables, $m\in\R$, and $s\in(0,1]$. Then $Y-m\in s\CA_\rho$ implies by monotoncity that $s^{-1}(X-m)\in\CA_\rho$, or equivalently, $X-m\in s\CA_\rho$. This suffices to prove $\rho_\star(X)\le\rho_\star(Y)$.
For cash-additivity, let $X\in L^\infty$, $k\in\R$, and observe that 
\begin{align*}\rho_\star(X+k)&=\{m\in\R\mid\exists\,s\in(0,1]:~ X+k-m\in s\CA_\rho\}\\
&=\{l\in\R\mid\exists\,s\in(0,1]:~X-l\in s\CA_\rho\}+k=\rho_\star(X)+k.
\end{align*}
Finally, in order to verify star-shapedness, let  $X\in L^\infty$ and $\lambda\in(0,1)$ and observe
\[\{\lambda\cdot m\mid m\in\R,\,\exists\,s\in(0,1]:~X-m\in s\CA_\rho\}\subset\{k\in\R\mid \exists\,s\in(0,1]:~\lambda X-k\in s\CA_\rho\},\]
which means $\rho_\star(\lambda X)\le\lambda\rho_\star(X)$.\\
For the first assertion in \eqref{eq:id1}, $\rho_\star(sY)\le 0$ holds for all $s\in(0,1]$ and all $Y\in\CA_\rho$, which means 
\[\tn{cl}\Big(\bigcup_{s\in(0,1]}s\CA_\rho\Big)\subset\CA_{\rho_\star}\]
by closedness of the right-hand set.
Conversely, suppose $\rho_\star(X)\le 0$, which means we must be able to find $(s_n)_{n\in\N}\in(0,1]$ such that $X-\tfrac 1 n\in s_n\CA_\rho$, $n\in\N$. Hence,
\[X=\lim_{n\to\infty}X-\tfrac 1 n\in\tn{cl}\Big(\bigcup_{s\in(0,1]}s\CA_\rho\Big).\]
For the second part of \eqref{eq:id1}, let  $\phi\in (L^\infty)^*$. Then
\begin{align*}\rho_\star^*(\phi)&=\sup_{R\in\CA_{\rho_\star}}\phi(R)=\sup_{Y\in\CA_\rho}\sup_{s\in(0,1]}s\phi(Y)=\sup_{s\in(0,1]}s\rho^*(\phi)=\rho^*(\phi).
\end{align*}
\item 
Suppose $X$ is dominated by $X'$ in the second-order stochastic dominance relation and $m\in\R$ is such that $X'-m=sY$ for some $s\in(0,1]$ and $Y\in\CA_\rho$. One verifies that $s^{-1}(X-m)$ is dominated by $s^{-1}(X'-m)=Y\in\CA_\rho$ in the second-order stochastic dominance relation. Hence, $X-m\in s\CA_\rho$ as well. We infer that 
\begin{align*}
\rho_\star(X)&=\inf\{k\in\R\mid \exists\,s\in(0,1]:~X-k\in s\CA_\rho\}\\
&\le\inf\{m\in\R\mid \exists\,s\in(0,1]:~X'-m\in s\CA_\rho\}=\rho_\star(X'),
\end{align*}
which proves consistency of $\rho_\star$.
\item In order to verify \eqref{eq:asymptoticequal}, note first that $\CA^\infty_\rho\subset\CA^\infty_{\rho_\star}$ follows from identity \eqref{eq:id1}.
Conversely, suppose that $U\in\CA_{\rho_\star}^\infty$, which means that for a null sequence $(s_n)_{n\in\N}\subset (0,1)$ and $(Y_n)_{n\in\N}\subset\CA_{\rho_\star}$, $U=\lim_{n\to\infty}s_nY_n$. Choosing $(t_n)_{n\in\N}\subset(0,1]$ and $(Z_n)_{n\in\N}\subset\CA_\rho$ appropriately, we can guarantee that $U=\lim_{n\to\infty}s_nt_nZ_n$. This is sufficient for $U\in\CA_\rho^\infty$. 
\item 
By \eqref{eq:id1},  $\dom(\rho^*_\star)=\dom(\rho^*)$. Hence, a probability density $Z^*\in L^1_+$ satisfies $\rho^*(Z^*)<\infty$ if and only if $\rho^*_\star(Z^*)<\infty$.  The assertion now follows from (3). 
\end{enumerate}
\end{proof}

\section{Local comonotone improvement}\label{sec:local}

\begin{definition}\label{def:locally comonotone}
Let $X\in L^\infty$ and $\pi\in\Pi$ be a finite measurable partition. 
A vector $\mbf Y\in\mathbb A_X$ is a \emph{locally comonotone allocation} of $X$ over $\pi$ if there are $(\mbf f^B)_{B\in\pi}\subset\mf C$ such that, for all $i\in[n]$,
\[Y_i=\sum_{B\in\pi} f_i^B(X)\ind_B.\]
\end{definition}

Given $\Q\ll\P$, in the sequel we denote by $\peq_\Q$ the \emph{$\Q$-convex order} on $L^\infty$: $X\peq_\Q Y$ holds if and only if, for all convex test functions $v\colon\R\to\R$, $\E_\Q[v(X)]\le\E_\Q[v(Y)]$.
\begin{lemma}\label{lem:1}
Suppose:
\begin{itemize}
\item[(i)] A vector $(\Q_1,\ldots,\Q_n)$ of probability measures satisfies Assumption~\ref{assum1}.
\item[(ii)] $\pi\in\Pi$ is a partition such that each $\frac{{\rm d}\Q_i}{{\rm d}\P}$ has a $\sigma(\pi)$-measurable version.
\item[(iii)]$\rho_i\colon L^\infty\to\R$ is a $\Q_i$-consistent risk measure, $i\in[n]$.
\end{itemize}
Let $X\in L^\infty$ and $\mbf X\in\mathbb A_X$ be arbitrary. Then there exists a locally comonotone allocation $\mbf Y\in \mathbb A_X$ over $\pi$ such that, for all $i\in[n]$,
\[Y_i\peq_{\Q_i}X_i\quad\text{and}\quad\rho_i(Y_i)\le \rho_i(X_i).\]
\end{lemma}
\begin{proof}
For each $B\in\pi$ consider the nonatomic probability space $(B,\CG_B,\P^{B})$, where 
$\CG_B:=\{B\cap A\mid A\in\CF\}$.
As 
\[\sum_{i=1}^nX_i|_{B}=X|_{B},\]
there is a comonotone $n$-partition of the identity $\mbf f^B\in\mf C$ such that $f_i^B(X|_{B})\peq_{\P^{B}}X_i|_{B}$ holds for all $i\in[n]$;
cf.\ \cite[Theorem 3.1]{Carlier}. 
In particular, setting 
\[Y_i:=\sum_{B\in\Pi}f_i^B(X)\ind_{B},\quad i\in[n],\]
defines an allocation $\mbf Y\in\mathbb A_X$ which we claim to be as described in (1). In order to verify this, recall that for each $i\in[n]$, 
\[\tfrac{{\rm d}\Q_i}{{\rm d}\P}=\sum_{B\in\pi}\tfrac{\Q_i(B)}{\P(B)}\ind_{B}\] 
by assumption (i).
Let $v:\R\to\R$ be an arbitrary convex function satisfying $v(0)=0$ without loss of generality, and compute
\begin{align*}
\E_{\Q_i}[v(Y_i)]&=\E_{\Q_i}\Big[\sum_{B\in\pi}v\big(f^B_i(X)\big)\ind_{B}\Big]=\sum_{B\in\pi}\Q_i(B)\E_{\P^{B}}\left[v\big(f_i^B(X|_{B})\big)\right]\\
&\le\sum_{B\in\pi}\Q_i(B)\E_{\P^{B}}[v(X_i|_{B})]=\E_{\Q_i}[v(X_i)].
\end{align*}
The same allocation $\mbf Y$ also satisfies $\rho_i(Y_i)\le\rho_i(X_i)$ for all $i\in[n]$ because each consistent risk measure $\rho_i$ is monotone with respect to $\peq_{\Q_i}$. 
\end{proof}

The previous lemma is the first step in the adaptation of the comonotone improvement procedure to our setting of heterogeneous reference probability measures. Under its assumptions, it implies in particular that for all vectors $\mbf X\in\prod_{i=1}^n\CA_{\rho_i}$ consisting of acceptable net losses, we find a locally comonotone allocation $\mbf Y$ of $\sum_{i=1}^n X_i$ over $\pi$ such that $\rho_i(Y_i)\le \rho_i(X_i)\le 0$, i.e., $\mbf Y\in\prod_{i=1}^n\CA_{\rho_i}$ as well. 

The next proposition provides the bounds necessary for the second step of comonotone improvement.

\begin{proposition}\label{prop:convergence}
Suppose:
\begin{itemize}
\item[(i)] A vector $(\Q_1,\ldots,\Q_n)$ of probability measures $\Q_i\approx\P$ satisfies Assumption~\ref{assum1}.
\item[(ii)] $\pi\in\Pi$ is a measurable partition of $\Omega$ as in \eqref{used once}.
\item[(iii)]$\rho_i\colon L^\infty\to\R$ is a $\Q_i$-consistent risk measure, $i\in[n]$.
\item[(iv)]The risk measures $\rho_1,\ldots,\rho_{n-1}$ are admissible and, for all $i\in[n-1]$ there is $Z_i\in\mbf C(\rho_i)$ such that 
\[\max_{k\in[n]}\rho_k^*(Z_i)<\infty.\]
\end{itemize}
Let $(Y_k)_{k\in\N}\subset\sum_{i=1}^n\CA_{\rho_i}$ be a bounded sequence and suppose that, for all $k\in\N$, $\mbf Y^k\in\prod_{i=1}^n\CA_{\rho_i}$ is a locally comonotone allocation of $Y_k$ over $\pi$. Then the sequence $(\mbf Y^k)_{k\in\N}$ is bounded as well. 
\end{proposition}
\begin{proof}
Let $\pi$ be the measurable partition of $\Omega$ from (ii). For $k\in\N$, there are $(\mbf f^{B,k})_{B\in\pi}\in\mf C$ such that, for all $i\in[n]$, 
\begin{equation}\label{locally comonotone}Y_i^k:=\sum_{B\in\pi}f^{B,k}_i(Y_k)\ind_{B}.\end{equation}
Let us first discuss boundedness of the sequence $(Y_1^k)_{k\in\N}$ in detail. Towards a contradiction, suppose $(Y_1^k)_{k\in\N}$ is unbounded in norm. 
As 
\begin{equation}\label{eq:limit1}Y^k_1=\sum_{B\in\pi}{\widetilde f}^{B,k}_1(Y_k)\ind_{B}+\sum_{B\in\pi}f_1^{B,k}(0)\ind_{B}\end{equation}
and the first summand is uniformly bounded in $k$ by Lipschitz continuity of $f_1^{B,k}$ and boundedness of the sequence $(Y_k)_{k\in\N}$, we obtain that 
$$\beta_k:=\max_{B\in\pi}|f_1^{B,k}(0)|=\Big\|\sum_{B\in\pi}f_1^{B,k}(0)\ind_{B}\Big\|_\infty$$
satisfies 
$\lim_{k\to\infty}\frac{\|Y_1^k\|_\infty}{\beta_k}=1.$
Moreover, up to passing to a subsequence, we can assume that for a suitable $\mbf u\in\R^\pi$ with $\max_{B\in\pi}|u_B|=1$,  
\begin{equation}\label{eq:contradiction}V:=\lim_{k\to\infty}\sum_{B\in\pi}\frac{f_1^{B,k}(0)}{\beta_k}\ind_{B}=\sum_{B\in\pi}u_B\ind_B\neq 0.\end{equation}
By definition of the asymptotic cone, $V\in\CA_{\rho_1}^\infty$. 
By Lemma~\ref{lem:direction}(1), for $Z_1\in\mbf C(\rho_1)$ chosen as in assumption (iv),
\begin{equation}\label{eq:leq}\E_\P[Z_1V]\le 0.\end{equation}
As $\rho_i^*(Z_1)<\infty$ for $i=2,...,n$ by (iv), we further obtain 
\begin{align*}
\E_\P[Z_1V]&=-\E_\P[Z_1(-V)]=-\lim_{k\to\infty}\sum_{i=2}^n\E_\P\left[Z_1\tfrac{Y_i^k}{\beta_k}\right]\ge-\lim_{k\to\infty}\tfrac 1{\beta_k}\sum_{i=2}^n\rho_i^*(Z_1)=0.
\end{align*}
Combining the latter estimate with \eqref{eq:leq}, $\E_\P[Z_1V]=0$ follows. 
Using that $V\in\CA_{\rho_1}^\infty$ and that $Z_1\in\mbf C(\rho_1)$, $V=0$ is a consequence in direct contradiction to \eqref{eq:contradiction}. The assumption that $(Y_1^k)_{k\in\N}$ is unbounded has to be absurd.

Arguing analogously for the coordinate sequences $(Y_i^k)_{k\in\N}$, $i=2,\ldots,n-1$, one infers their norm boundedness. 
At last,
\[\sup_{k\in\N}\|Y_n^k\|_\infty\le \sup_{k\in\N}\|Y_k\|_\infty+\sum_{i=1}^{n-1}\sup_{k\in\N}\|Y_i^k\|_\infty<\infty.\] 
\end{proof}

The most important consequence of Proposition~\ref{prop:convergence} in the sequel is the following corollary establishing closedness of the Minkowski sum of individual acceptance sets. 

\begin{corollary}\label{cor:convergence}
In the situation of Proposition~\ref{prop:convergence}, the Minkowski sum $\sum_{i=1}^n\CA_{\rho_i}$ is closed. 
\end{corollary}
\begin{proof}
Abbreviate $\CA_+:=\sum_{i=1}^n\CA_{\rho_i}$ and let $(Y_k)_{k\in\N}\subset\CA_+$ be a sequence that converges to $Y\in\Linfty$ as $k\to\infty$.
We need to prove that $Y\in\CA_+$.
To this effect, for all $k\in\N$ let $\mbf Y^k$ be a locally comonotone allocation as in \eqref{locally comonotone} such that, for all $i\in[n]$, $Y_i^k\in\CA_{\rho_i}$. 
By Proposition~\ref{prop:convergence},
$\sup_{k\in\N}\max_{i\in[n]}\|Y_i^k\|_\infty<\infty.$
By \eqref{eq:limit1}
\[\kappa:=\sup\big\{|f^{B,k}_i(0)|\mid k\in\N,\,B\in\pi,\,i\in[n]\big\}\le\sup\{\|Y^k\|_\infty+\max_{i\in[n]}\|Y_i^k\|_\infty\mid k\in\N\}<\infty.\]
The set $\mf C_\kappa:=\{\mbf f\in\mf C\mid \max_{i\in[n]}|f_i(0)|\le \kappa\}$ is sequentially compact in the topology of pointwise convergence; cf.\ \cite[Lemma B.1]{Svindland}.
Hence, there are $(\mbf g^B)_{B\in\pi}\subset\mf C_\kappa$ such that, up to switching to a subsequence $|\pi|$ times, $\lim_{k\to\infty}\mbf f^{B,k}(x)=\mbf g^B(x)$ for all $B\in\pi$ and all $x\in\R$.
For each $i\in[n]$, we further observe that
\begin{align*}\big\|\sum_{B\in\pi}(f_i^{B,k}(Y^k)-g^B_i(Y))\ind_{B}\big\|_\infty&\le\sum_{B\in\pi}\big\|\big(f_i^{B,k}(Y^k)-g^B_i(Y)\big)\ind_{B}\big\|_\infty\\
&\le\sum_{B\in\pi}\big\|(\widetilde f_i^{B,k}-\widetilde g^B_i)(Y^k)\big\|_\infty+\|Y^k-Y\|_\infty+|f^{B,k}_i(0)-g^B_i(0)|.
\end{align*}
As $\widetilde{\mbf f}^{B,k}$ converges to $\widetilde{\mbf g}^B$ uniformly on the compact interval $[-\sup_{k\in\N}\|Y_k\|_\infty,\sup_{k\in\N}\|Y_k\|_\infty]$, we infer that 
$\lim_{k\to\infty}\big\|\sum_{B\in\pi}(f_i^{B,k}(Y_k)-g^B_i(Y))\ind_B\big\|_\infty=0$.
As $\CA_{\rho_i}$ is closed, 
$\sum_{B\in\pi}g^B_i(Y)\ind_{B}\in\CA_{\rho_i}$ must hold, $i\in[n]$. 
It remains to observe that the latter defines a vector in $\mathbb A_Y$ and that therefore $Y\in\CA_+$. 
\end{proof}

\begin{remark}\label{Hans}
Proposition~\ref{prop:convergence} and Corollary~\ref{cor:convergence} are the key to verifying the existence of optimal allocations claimed in Theorem~\ref{thm:risksh_mon}. 
The role that Proposition~\ref{prop:convergence}---and thereby admissibility of individual risk measures---plays for Corollary~\ref{cor:convergence} and Theorem~\ref{thm:risksh_mon} is best understood as adaptation of the spirit of Dieudonn\'e's \cite{Dieudonne} famous theorem to the present nonconvex situation. 
The proof of Proposition~\ref{prop:convergence} justifies the reliance on Assumption~\ref{assum1} of our approach based on comonotone improvement.
\textit{A priori}, the comonotone functions $(\mbf f^B)_{B\in\pi}$ which describe the allocation on events $B\in\pi$ have no apparent relation to each other. 
This problem is circumvented by the use that equation~\eqref{eq:contradiction} makes of the compactness of the finite-dimensional unit ball in $\R^\pi$. 
We do not see an immediate generalisation to infinite dimensions, i.e., completely arbitrary heterogeneous reference measures.
\end{remark}

\section{Proofs from Sections~\ref{sec:consistent} and~\ref{sec:heterogeneous}}\label{proofs of sec 2}

\begin{proof}[Proof of Lemma~\ref{lem:is equiv}]\
\begin{enumerate}[(1)]
    \item Assume $\Q\not\approx\P$ and fix $Y\in\CA_\rho$ as well as a $\Q$-null set $N\in\CF$ which satisfies $\P(N)>0$. 
For all $r\in\R$, $Y+r\ind_N\sim_\Q Y$, whence $Y+\tn{span}\{\ind_N\}\subset\CA_\rho$ follows.
If $Z^*\in L^1$ satisfies $\rho^*(Z^*)<\infty$, \[\sup_{r\in\R}\E_\P[Z^*(Y+r\ind_N)]\le \rho^*(Z^*)\]
must hold. This entails $\E_\P[Z^*\ind_{N}]=0$. However, we also observe that $\ind_N=\lim_{k\to\infty}\tfrac 1 k(Y+k\ind_N)$
lies in the asymptotic cone $\CA_\rho^\infty$. Thus, no $Z^*\in \dom(\rho^*)\cap L^1$ can satisfy requirements (a)--(b) in Definition~\ref{def:compatible}.
\item Let $Z^*\in\mbf C(\rho)$ and note that $-\ind_{\{Z^*=0\}}\in\CA_\rho^\infty$. As $\E_\P[Z^*(-\ind_{\{Z^*=0\}})]=0$, $\P(Z^*=0)=0$ follows. 
\item Let $Q,Z$, and $\lambda$ be as described and abbreviate 
$R:=\lambda Q+(1-\lambda)Z.$
For property (a) in Definition~\ref{def:compatible}, convexity of $\dom(\rho^*)$ implies $R\in\dom(\rho^*)$.
As for property (b), suppose $U\in\CA_\rho^\infty$ satisfies $\E_\P\left[RU\right]=0$. By Lemma~\ref{lem:direction}(1), both $\E_\P[QU]\le 0$ and $\E_\P[ZU]\le 0$. This forces $\E_\P[QU]=\E_\P[ZU]=0$. As $Q\in\mbf C(\rho)$, $U=0$ follows.
\end{enumerate}
\end{proof}

For reasons that will become clear momentarily, we shall first give the proof of Proposition~\ref{uncountably many}.
\begin{proof}[Proof of Proposition~\ref{uncountably many}]
Clearly, if $\rho$ is admissible, the identity $\rho=\E_\Q[\cdot]$ is impossible. 
Conversely, star-shapeness of $\rho$ implies the following substantially simpler characterisation of the asymptotic cone:
\[\CA_\rho^\infty=\{U\in L^\infty\mid \forall\,t\ge 0:~tU\in\CA_\rho\}.\]  
By \cite[Theorem 5.7]{Collapse},
$\rho\neq\E_\Q[\cdot]$ if and only if, for all $U\in \CA_\rho^\infty$, $\E_\Q[U]=0$ only holds if $U=0$. The latter is equivalent to $\frac{{\rm d}\Q}{{\rm d}\P}\in\mbf C(\rho)$ and implies admissibility of $\rho$. 
\end{proof}

\begin{proof}[Proof of Proposition~\ref{compatible nonempty}]
(1) implies (4): $\rho$ is admissible if and only if the star-shaped hull $\rho_\star$ introduced in \eqref{starshaped hull} is admissible.
By \cite[Theorem 5.7]{Collapse} and \eqref{eq:id1}, the latter is the case if and only if $\dom(\rho^*)\cap L^1=\dom(\rho_\star^*)\cap L^1$ contains at least two elements. 

(3) implies (2): If (3) holds, then $\rho_\star=\E_\Q[\cdot]$ is not possible. 
By the argument from the proof of Proposition~\ref{uncountably many} and Lemma~\ref{lem:star}, $\frac{{\rm d}\Q}{{\rm d}\P}\in\mbf C(\rho_\star)=\mbf C(\rho)$. 

The implications (4) implies (3) and (2) implies (1) are trivial.

(2) implies (5): Suppose that $\rho\le\E_\Q[\cdot]+\beta$ for some $\beta>0$. Then for all $U\in L^\infty$ with $\E_\Q[U]\le 0$, $nU-\beta\in\CA_\rho$. Thus, $U=\lim_{n\to\infty}\tfrac 1 n(nU-\beta)\in\CA_\rho^\infty$. This observation clearly precludes $\tfrac{{\rm d}\Q}{{\rm d}\P}\in\mbf C(\rho)$. 
\end{proof}

We now turn to the results from Section~\ref{sec:heterogeneous}. 

\begin{lemma}\label{lem:partition}
A function $\ph\colon L^\infty\to\R$ is dilatation monotone above a finite $\sigma$-algebra $\CH$ if and only if $\ph$ is dilatation monotone above $\sigma(\pi)$ for some $\pi\in\Pi$.
\end{lemma}
\begin{proof}
We only have to prove necessity. Let $\CH$ be the finite sub-$\sigma$-algebra in question. The set $\Pi(\CH):=\{\pi\in\Pi\mid \pi\subset\CH\}$ is nonempty as $\{\Omega\}\in\Pi(\CH)$. Select $\pi^*\in \Pi(\CH)$ with maximal cardinality and prove that every $A\in\CH$ has to satisfy 
$\ind_A=\ind_E$ in $L^\infty$ for some $E\in\sigma(\pi^*)$.
Hence, $\sigma(\pi^*)\subset\CH\subset\sigma(\pi^*\cup\CN)$, where $\CN$ denotes the collection of all $\P$-null sets. As dilatation monotonicity of $\ph$ above $\sigma(\pi^*)$ is equivalent to dilatation monotonicity above $\sigma(\pi^*\cup\CN)$, the proof is concluded.
\end{proof}

\begin{proof}[Proof of Theorem~\ref{thm:Pbased}]
(1) implies (2): If probability measures $(\Q_1,\ldots,\Q_n)$ satisfy Assumption~\ref{assum1}, we may fix a vector $\mbf Q$ of versions $Q_i$ of $\tfrac{{\rm d}\Q_i}{{\rm d}\P}$ which are simple functions. The set 
$$\Sigma:=\{\mbf q\in\R^n\mid\P(\mbf Q=\mbf q)>0\}$$
is finite.
The event $\bigcap_{\mbf q\in\Sigma}\{\mbf Q=\mbf q\}^c$ is a null set which we can assume to be empty by suitably manipulating the choice of densities. Hence, setting $\pi:=\{\{\mbf Q=\mbf q\}\mid \mbf q\in\Sigma\}\in\Pi$, each density has a $\sigma(\pi)$-measurable version. Set $\mathcal H:=\sigma(\pi)$, let $\CG\supset\CH$ be another sub-$\sigma$-algebra of $\CF$, and note that for all $i\in[n]$ and all $X\in L^\infty$ (by the Bayes rule for conditional expectations) 
\[
\tfrac{{\rm d}\Q_i}{{\rm d}\P}\E_{\Q_i}[X|\CG]=\E_\P[\tfrac{{\rm d}\Q_i}{{\rm d}\P}X|\CG]=\tfrac{{\rm d}\Q_i}{{\rm d}\P}\E_\P[X|\CG].
\]
In view of $\tfrac{{\rm d}\Q_i}{{\rm d}\P}>0$ by equivalence, $\E_{\Q_i}[X|\CG]=\E_\P[X|\CG]$. As moreover each $\Q_i$-consistent risk measure is $\Q_i$-dilatation monotone, we have 
\[\rho_i\left(\E_\P[X|\CG]\right)=\rho_i\left(\E_{\Q_i}[X|\CG]\right)\le\rho_i(X).\]
Each $\rho_i$ is dilatation monotone above $\CH$. 

(2) implies (1): Suppose that each $\rho_i$ is dilatation monotone above a finite sub-$\sigma$-algebra $\CH\subset\CF$.
By Lemma~\ref{lem:partition}, we can assume without loss of generality that $\CH=\sigma(\pi)$ for some $\pi\in\Pi$. The further proof proceeds in several steps.

\textsf{Step 1:} Fix $i\in[n]$, $E\in\pi$, $Y\in L^\infty$, and consider the continuous function $\ph\colon L^\infty_{\P^E}\to\R$ defined by 
\[\ph(X):=\rho_i\left(Y\ind_{E^c}+X\ind_{E}\right).\footnote{~$L^\infty_{\P^E}$ is isometrically isomorphic to $\{X\ind_E\mid X\in L^\infty\}$.}\]
Then $\ph$ is $\P^E$-law invariant. 
By \cite[Theorem 1.1]{Cherny}, it suffices to show dilatation monotonicity. Let $\CG\subset\CF$ be an arbitrary sub-$\sigma$-algebra and consider the sub-$\sigma$-algebra 
$\CG_E:=\{(A\cap E)\cup B\mid A\in\CG,\,B\in\CF,\,B\subset E^c\}\supset\CH$. 
Note that, for all $A\in\CG$ and all $B\in\CF$ with $B\subset E^c$, 
\begin{align*}
\E_\P[(X\ind_{E}+Y\ind_{E^c})\ind_{(A\cap E)\cup B}]&=\E_\P[X\ind_{A\cap E}]+\E_\P[Y\ind_B]=\P(E)\E_{\P^{E}}[X\ind_A]+\E_\P[Y\ind_B]\\
&=\P({E})\E_{\P^{E}}\left[\E_{\P^{E}}[X|\CG]\ind_A\right]+\E_\P[Y\ind_B]\\
&=\E_\P\left[\E_{\P^{E}}[X|\CG]\ind_{A\cap {E}}+Y\ind_B\right]=\E_\P\left[(\E_{\P^{E}}[X|\CG]\ind_{E}+Y\ind_{E^c})\ind_{(A\cap E)\cup B}\right].
\end{align*}
As $\E_{\P^{E}}[X|\CG]\ind_{E}+Y\ind_{E^c}$ can also be verified to be $\CG_E$-measurable, we obtain 
\[\E_\P[X\ind_{E}+Y\ind_{E^c}|\CG_E]=\E_{\P^{E}}[X|\CG]\ind_{E}+Y\ind_{E^c}.\]
Hence, by dilatation monotonicity of $\rho_i$, 
\[\ph(\E_{\P^{E}}[X|\CG])=\rho_i\left(\E_\P[X\ind_{E}+Y\ind_{E^c}|\CG_E]\right)\le \rho_i(X\ind_{E}+Y\ind_{E^c})=\ph(X).\]
\textsf{Step 2:} If $\Q_i^E=\P^E$ for all $E\in\pi$, then \eqref{used once} holds automatically. 

\textsf{Step 3:} Assume otherwise that
there is $i\in[n]$ and $E\in\pi$ such that $\P^E\neq \Q_i^E$. Let $\es_1$ be the Expected Shortfall at level 1 computed under $\P$; cf.\ \eqref{eq:ES} below. 
Then
$\rho_i=\es_1$
must hold. 
As both $\rho_i$ and $\es_1$ are $\Norm_\infty$-continuous, it suffices to verify the identity for an arbitrary simple random variable. $X\le \es_1(X)$ holds for all $X\in L^\infty$. By monotonicity and cash-additivity, $\rho_i(X)\le\rho_i\left(\es_1(X)\right)=\es_1(X)$, and this estimate is an equality if $X$ is constant $\P$-a.s.\ Assume now that we find a simple random variable $X$ such that $\rho_i(X)<\es_1(X)$. 
Using the Dubins-Spanier Theorem~\cite[Theorem 13.34]{Ali} we can switch to $X'\sim_{\Q_i}X$ if necessary and assume that $X$ attains all its finitely many possible values on $E$ with positive probability. 

The functional $L^\infty_{\P^E}\ni V\mapsto \rho_i(X\ind_{E^c}+V\ind_E)$ is law invariant both with respect to $\P^E$ and to $\Q^E_i$; has the Fatou property by \cite[Theorem 3.5]{Consistent};
and is monotone.
Note that $-\es_1(-X)<\es_1(X)$, and observe that $Y\in L^\infty$ defined by
\[Y:=\begin{cases}X&\text{on }E^c,\\
-\es_1(-X)&\text{on }\{X<\es_1(X)\}\cap E,\\
\es_1(X)&\text{on }\{X=\es_1(X)\}\cap E,\end{cases}\]
satisfies $Y\le X$.
Hence, $\rho_i(Y)<\es_1(X)$. 
Define 
\[p^*:=\sup\{\P^E(A)\mid A\in\CF,\,A\subset E,\,\rho_i(X\ind_{E^c}+\es_1(X)\ind_A-\es_1(-X)\ind_{E\setminus A}=\rho_i(Y)\}.\]
We have $p^*\ge \P^E(X<\es_1(X))>0$. Moreover, the Fatou property implies for all $A\in\CF$ with $\P^E(A)=p^*$ that $\rho_i(X\ind_{E^c}+\es_1(X)\ind_A-\es_1(-X)\ind_{E\setminus A})=\rho_i(Y)$, whence also $p^*<1$. 

By, e.g., \cite[Lemma A.2]{Collapse}, we may choose $A,B\in\CF$ such that $A,B\subset E$ and 
\[\Q_i^E(A)<\P^E(A)=p^*=\P^E(B)<\Q_i^E(B).\]
Use the convex range property of $\Q_i^E$ to select $C\in\CF$ such that $A\subset C\subset E$ and $\Q_i^E(B)=\Q_i^E(C)$. In particular, $\Q_i^E(C\setminus A)>0$ implies $\P^E(C\setminus A)>0$, whence $\P^E(C)>p^*$ follows. By definition of $p^*$ and monotonicity of $\rho_i$, 
\begin{equation}\label{first contr}\rho_i\big(X\ind_{E^c}+\es_1(X)\ind_C-\es_1(-X)\ind_{E\setminus C}\big)>\rho_i(Y).\end{equation}
As $X\ind_{E^c}+\es_1(X)\ind_C-\es_1(-X)\ind_{E\setminus C}\sim_{\Q_i}X\ind_{E^c}+\es_1(X)\ind_B-\es_1(-X)\ind_{E\setminus B}$ though and $\P^E(B)=p^*$, 
\begin{equation}\label{second contr}\rho_i\big(X\ind_{E^c}-\es_1(-X)\ind_C+\es_1(X)\ind_{E\setminus C}\big)=\rho_i(Y).\end{equation}
\eqref{first contr} and \eqref{second contr} together yield a contradiction. The assumption $\rho_i(X)<\es_1(X)$ for some simple $X$ must be absurd. 

\textsf{Step 4:} If $\rho_i=\es_1$, it is naturally a $\P$-consistent risk measure. We can thus replace $\Q_i$ in the vector $(\Q_1,\ldots,\Q_n)$ of reference measures by $\P$. 
The potentially modified vector $(\Q_1,\ldots,\Q_n)$ satisfies \eqref{used once}, the alternative shape \eqref{used once} of Assumption~\ref{assum1}.
\end{proof}

\begin{proof}[Proof of Proposition~\ref{prop:sufficient}]
If $\pi\in\Pi$ is such that every density $\tfrac{{\rm d}\Q_i}{{\rm d}\Q_1}$ has a $\sigma(\pi)$-measurable version, then 
\[\E_{\Q_i}[X|\sigma(\pi)]=\sum_{B\in\pi}\frac{\E_{\Q_i}[X\ind_B]}{\Q_i(B)}\ind_B=\sum_{B\in\pi}\E_{\Q_1^B}[X]\ind_B.\]
The latter is independent of $i$. 

Conversely, we can argue as  in Lemma~\ref{lem:partition} 
to find $\pi\in\Pi$ such that $\sigma(\pi)$ is sufficient for $\{\Q_1,\ldots,\Q_n\}$. In that case, for $A\in\CF$,
\[\Q_i(A)=\E_{\Q_i}[\E_{\Q_i}[\ind_A|\sigma(\pi)]]=\E_{\Q_i}\left[\E_{\Q_1}[\ind_A|\sigma(\pi)]\right]=\sum_{B\in\pi}\Q_i(B)\Q_1^B(A).\]
As \eqref{used once} holds, also Assumption~\ref{assum1} must hold.
\end{proof}

\section{Proofs of Theorem~\ref{thm:risksh_mon}, Corollary~\ref{cor:problem2}, and Corollary~\ref{cor:Pbased}}\label{proofs of sec 3}

\begin{proof}[Proof of Theorem~\ref{thm:risksh_mon}]
First of all, we prove that $\rho$ does not attain the value $-\infty$.
To this end, let $Z_1\in\mbf C(\rho_1)$ be as described in assumption (iii), let $X\in\Linfty$, and fix $\mbf Y\in\mathbb A_X$.
By definition of the convex conjugate,
\[\sum_{i=1}^n\rho_i(Y_i)\ge\sum_{i=1}^n\E_\P[Z_1Y_i]-\rho_i^*(Z_1)=\E_\P[Z_1X]-\sum_{i=1}^n\rho_i^*(Z_1).\]
Taking the infimum over $\mbf Y\in\mathbb A_X$ on the left-hand side proves
\[\rho(X)\ge\E_\P[Z_1X]-\sum_{i=1}^n\rho_i^*(Z_1)>-\infty.\]
Also, the Minkowski sum 
$\CA_+:=\sum_{i=1}^n\CA_{\rho_i}$ is closed by Corollary~\ref{cor:convergence}.
The remainder of the proof is standard, but shall be included for the sake of completeness.
In its next step, we verify $\CA_+=\CA_\rho$. The inclusion of the left-hand set in the right-hand set is immediate. Conversely, assume without loss of generality $\rho(X)=0$ and consider a sequence $(\mbf Y^k)_{k\in\N}\subset\mathbb A_X$ such that 
$\lim_{k\to\infty}\sum_{i=1}^n\rho_i(Y_i^k)=0$.
Let $U^k_i:=Y^k_i-\rho_i(Y^k_i)\in\CA_{\rho_i}$, $i\in[n]$. As $\CA_+$ is closed, $X=X-\lim_{k\to\infty}\sum_{i=1}^n\rho_i(Y_i^k)=\lim_{k\to\infty}\sum_{i=1}^nU^k_i\in\CA_+$. 
Now, let $X\in L^\infty$. As $\rho\big(X-\rho(X)\big)=0$, $X-\rho(X)\in\CA_+$ and we can find $\mbf Y\in\prod_{i=1}^n\CA_{\rho_i}$ such that $X-\rho(X)=\sum_{i=1}^nY_i$. Defining an allocation $\mbf X\in\mathbb A_X$ by $X_i:=Y_i+\tfrac 1 n\rho(X)$, $i\in[n]$, we obtain
\[\rho(X)\le\sum_{i=1}^n\rho_i(X_i)=\sum_{i=1}^n\rho_i\big(Y_i+\tfrac 1 n\rho(X)\big)=\sum_{i=1}^n\rho_i(Y_i)+\rho(X)\leq \rho(X).\] 
Hence, $\mbf X$ is an optimal allocation of $X$.
\end{proof}

\begin{proof}[Proof of Corollary~\ref{cor:problem2}]
Let $(\mbf f^{B,k})_{B\in\pi,k\in\N}\subset\mf C$ such that, for all $k\in\N$,
\[r:=(\Box_{i=1}^n\rho_i)(X)=\sum_{i=1}^n\rho_i\Big(\sum_{B\in\pi}f^{B,k}_i(X)\ind_B\Big),\]
and such
$\Xi\big((\mbf f^{B,k})_{B\in\pi}\big)$ converges to the optimal value of \eqref{problem2} as $k\to\infty$.
Define $\mbf g^{B,k}\colon\R\to\R^n$ by \[g_i^{B,k}(x)=f_i^{B,k}(x+r)-\tfrac r n,\quad x\in\R,\]
and note that $(\mbf g^{B,k})_{B\in\pi,k\in\N}\subset\mf C$ as well as 
\[\sum_{i=1}^n\rho_i\Big(\sum_{B\in\pi}g_i^{B,k}(X-r)\Big)=0,\quad i\in[n].\]
As in the proof of Corollary~\ref{cor:convergence}, 
$\kappa:=\sup\big\{|g^{B,k}_i(0)|\mid k\in\N,\,B\in\pi,\,i\in[n]\big\}<\infty$.
Notice that this implies that 
\[\sup\big\{|f^{B,k}_i(0)|\mid \mid k\in\N,\,B\in\pi,\,i\in[n]\big\}\le \kappa+2|r|<\infty.\]
Using the compactness argument from the proof of Corollary~\ref{cor:convergence}, we find a solution $(\mbf f^B)_{B\in\pi}$ to \eqref{problem2}. 
\end{proof}

\begin{proof}[Proof of Corollary~\ref{cor:Pbased}]
Let $X,Y\in L^\infty$ and suppose that $X\sim_{\P^B}Y$ holds for all $B\in\pi$. Moreover, let $(\mbf f^B)_{B\in\pi}\subset\mf C$ and denote the resulting locally comonotone partitions of $X$ and $Y$ by $\mbf X$ and $\mbf Y$, respectively. We compute for any $i\in[n]$ and $x\in\R$: 
\begin{align*}\Q_i(X_i\le x)&=\sum_{B\in\pi}\Q_i(B)\P^B\left(f^B_i(X)\le x\right)=\sum_{B\in\pi}\Q_i(B)\P^B\left(f^B_i(Y)\le x\right)=\Q_i(Y_i\le x).
\end{align*}
That is, $X_i\sim_{\Q_i}Y_i$ holds for all $i\in[n]$.
Given the $\Q_i$-law-invariance of each $\rho_i$ and the proof of Theorem~\ref{thm:risksh_mon},
\begin{align*}\rho(X)&=\min_{(\mbf f^B)_{B\in\pi}\subset\mf C}\sum_{i=1}^n\rho_i\Big(\sum_{\pi\in B}f^B_i(X)\ind_B\Big)=\min_{(\mbf f^B)_{B\in\pi}\subset\mf C}\sum_{i=1}^n\rho_i\Big(\sum_{\pi\in B}f^B_i(Y)\ind_B\Big)=\rho(Y).
\end{align*}
\end{proof}

\section{Proof of Theorem~\ref{thm:application2}}\label{last appendix}

\begin{lemma}\label{lem:sumasymptotic}
In the situation of Theorem~\ref{thm:application2} set 
$\rho:=\Box_{i=1}^n\rho_i$ and $\tau_i:=(\rho_i)_\star$, $i\in[n]$.
Then 
\begin{equation}\label{outcome}\CA_\rho^\infty\subset\sum_{i=1}^n\CA_{\tau_i}^\infty=\sum_{i=1}^n\CA_{\rho_i}^\infty.\end{equation}
\end{lemma}
\begin{proof}
The second identity in \eqref{outcome} is a direct consequence of \eqref{eq:asymptoticequal}.
For the first inclusion, let $U\in\CA_\rho^\infty$, allowing us to find a null sequence $(s_k)_{k\in\N}\subset(0,1)$ and $(Y_k)_{k\in\N}\subset\CA_\rho$ such that $\lim_{k\to\infty}s_kY_k=U$. By Corollary~\ref{cor:convergence} and Theorem~\ref{thm:risksh_mon}, 
$\CA_{\rho}=\sum_{i=1}^n\CA_{\rho_i}$, whence one conludes for all $k\in\N$ that $s_kY_k\in\sum_{i=1}^n\CA_{\tau_i}$.
Each risk measure $\tau_i$ is $\Q_i$-consistent, $i\in[n]$, by Lemma~\ref{lem:star}(2), and the assumptions of Lemma~\ref{lem:1} are met. We can thus find a family $(\mbf f^{B,k})_{B\in\pi,\,k\in\N}\subset \mf C$, such that the associated locally comonotone allocation over $\pi$ satisfies 
\[\sum_{B\in\pi} f^{B,k}_i(s_kY_k)\ind_B\in \CA_{\tau_i},\quad i\in[n].\]
The consistent risk measures $\tau_1,\ldots,\tau_n$ are even admissible by Lemma~\ref{lem:star}(4).
Applying the argument from the proof of Corollary~\ref{cor:convergence}, there is a subsequence $(k_\lambda)_{\lambda\in\N}$ and $(\mbf g^B)_{B\in\pi}\subset\mf C$ such that, for all $i\in[n]$,
\[
U_i:=\sum_{i=1}^ng_i^B(U)\ind_B=\lim_{\lambda\to\infty}\sum_{B\in\pi}f^{B,k_\lambda}_i(s_{k_\lambda}Y_{k_\lambda})\ind_B.
\]
Hence, $\mbf U\in\mathbb A_U$ and  $U\in\sum_{i=1}^n\CA_{\tau_i}$.

As $\CA_\rho^\infty$ is a cone, we can repeat this argument for each multiple $mU$, $m\in\N$, and obtain a family $(\mbf g^{B,m})_{B\in\pi,\,m\in\N}$ in $\mf C$ such that, for all $m\in\N$ and $i\in[n]$, 
$\sum_{B\in\pi}g_i^{B,m}(mU)\ind_B\in\CA_{\tau_i}$.
In particular, for $\w{\mbf g}^{B,m}:=\tfrac 1 m\mbf g^{B,m}(m\,\cdot)\in\mf C$, star-shapedness of $\CA_{\tau_i}$ yields, for all $i\in[n]$,
$\sum_{B\in \pi}\w{\mbf g}^{B,m}(U)\ind_B\in\CA_{\tau_i}$.
Using the above argument once more, there is a subsequence $(m_\lambda)_{\lambda\in\N}$ and $(\mbf h^B)_{B\in\pi}\subset\mf C$ such that, for all $i\in[n]$,
\[\lim_{\lambda\to\infty}\sum_{B\in \pi}\w g^{B,m_\lambda}_i(U)\ind_B=\sum_{B\in \pi} h_i^{B}(U)\ind_B\in\CA_{\tau_i}^\infty.\]
This is sufficient to show that 
\[U=\sum_{i=1}^n\sum_{B\in\pi}h_i^B(U)\ind_B\in\sum_{i=1}^n\CA_{\tau_i}^\infty.\]
\end{proof}

\begin{proof}[Proof of Theorem~\ref{thm:application2}]
First, we prove that $\eta(X)>-\infty$ holds for all $X\in L^\infty$.
Let $Z^*$ be the probability density from Assumption~\ref{assum:RMR}. For all $\mbf X\in\mathbb A_X$, 
\begin{align*}
\sum_{i=1}^n\eta_i(X_i)&=\sum_{i=1}^n\inf\{\E_\P[Z^*S_i]\mid S_i\in\mathcal S_i,~X_i-S_i\in\CA_{\rho_i}\}\\
&=\sum_{i=1}^n\inf\{\E_\P[Z^*X_i]-\E_\P[Z^*(X_i-S_i)]\mid S_i\in\mathcal S_i,~X_i-S_i\in\CA_{\rho_i}\}\\
&\ge \sum_{i=1}^n\inf\{\E_\P[Z^*X_i]-\E_\P[Z^*Y_i]\mid Y_i\in\CA_{\rho_i}\}= \E_\P[Z^*X]-\sum_{i=1}^n\rho_i^*(Z^*).
\end{align*}
Recalling that the risk measure $\rho:=\Box_{i=1}^n\rho_i$ satisfies $\CA_\rho=\sum_{i=1}^n\CA_{\rho_i}$, 
we shall now verify that $\CA_\rho+\ker(\pi)$ is closed.
Let $(A^k)_{k\in\N}\subset\CA_\rho$ and $(N^k)_{k\in\N}\subset \ker(\pi)$ be sequences such that $Y^k:=A^k+N^k\to Y\in L^\infty$ as $k\to\infty$. Assume towards a contradiction that the sequence $(A^k)_{k\in\N}$ is unbounded.
This also implies that $(N^k)_{k\in\N}$ is unbounded and that 
$\lim_{k\to\infty}\frac{\|N^k\|_\infty}{\|A^k\|_\infty}=1.$
As $\ker(\pi)$ is of finite dimension, we can assume up to potentially switching to a subsequence that, for suitable $N^*\in\ker(\pi)$, 
\[
\lim_{k\to\infty}\frac{N^k}{\|A^k\|_\infty}=N^*=-\lim_{k\to\infty}\frac{A^k}{\|A^k\|_\infty}.
\]
By definition, $-N^*\in\CA_\rho^\infty$. Set $\tau_i:=(\rho_i)_\star$ as in \eqref{starshaped hull}. By Lemma~\ref{lem:sumasymptotic}, we can find $K_i\in\CA_{\tau_i}^\infty$, $i\in[n]$, such that $\mbf K\in\mathbb A_{-N^*}$.
Using Lemma~\ref{lem:direction}(1) in the first and second estimate, we find for all $j\in[n]$ that 
\[0\ge \E_\P[Z^*K_j]\ge\sum_{i=1}^n\E_\P[Z^*K_i]=\E_\P[Z^*(-N^*)]=0.\]
As $Z^*\in\bigcap_{i=1}^n\mbf C(\tau_i)$ by Lemma~\ref{lem:star}(3), $K_i=0$ for all $i\in[n]$, whence $N^*=0$ follows.
The assumption that $(A^k)_{k\in\N}$ is unbounded has to be absurd.
Consequently, the sequence $(N^k)_{k\in\N}=(Y^k-A^k)_{k\in\N}$ is also bounded. As $\ker(\pi)$ is of finite dimension and $\CA_\rho$ is closed, there are $N\in\ker(\pi)$ and $A\in\CA_\rho$ such that, for a suitably chosen subsequence $(k_\lambda)_{\lambda\in\N}$, $A^{k_\lambda}\to A$ and $N^{k_\lambda}\to N$, $\lambda\to\infty$. We have verified that 
$$Y=A+N\in\CA_\rho+\ker(\pi).$$
Let $X\in\sum_{i=1}^n\dom(\eta_i)=\dom(\eta)$ and $\mbf U$, $U$ as in the assertion. For all $k\in\N$ we may choose an allocation $\mbf X^k\in\mathbb A_X$ such that 
\[\sum_{i=1}^n\eta_i(X_i^k)\le \eta(X)+\tfrac 1 k.\]
Moreover, for all $i\in[n]$ we may select $S^k_i\in\mathcal S_i$ such that 
\[\price_i(S^k_i)\le \eta(X_i^k)+\tfrac 1{kn}\quad\text{and}\quad X^k_i-S^k_i\in\CA_{\rho_i}.\]
For the global security $S^k:=\sum_{i=1}^nS_i^k\in\mathcal M$, we observe
\[X-\pi(S^k)U=(X-S^k)+(S^k-\pi(S^k)U)\in\CA_\rho+\ker(\pi).\]
As $\lim_{k\to\infty}\pi(S^k)=\eta(X)$, 
$X-\eta(X)U\in\CA_\rho+\ker(\pi)$
holds by closedness of the latter Minkowski sum. For appropriate $\mbf N\in\prod_{i=1}^n\Scal_i$, setting $N:=\sum_{i=1}^nN_i\in\ker(\pi)$,
$X-\eta(X)U-N\in\CA_\rho$. Using Lemma~\ref{lem:1}, there is a family $(\mbf f^B)_{B\in\pi}\subset\mf C$ such that, for all $i\in[n]$, 
$\sum_{B\in\pi}f_i^B(X-\eta(X)U-N)\ind_B\in\CA_{\rho_i}$.
\eqref{eq:shape} is now an immediate consequence of $\Scal_i$-additivity of $\rho_i$, i.e., for all choices of $X\in\Linfty$ and $S\in\Scal_i$, $\rho_i(X+S)=\rho_i(X)+\price_i(S)$. 
\end{proof}

\begin{proof}[Proof of Proposition~\ref{not li}]
(1) clearly implies (2). For the converse, we shall assume towards a contradiction that $\mathcal S$ contains a nonconstant random variable and that $\eta$ is $\P$-law invariant. 
We then first observe that 
\begin{equation}\label{only trivial}\dom(\eta^*)=\{1\}.\end{equation}
To this end, note that, for all $X\in\dom(\eta)$,
\begin{align}\label{conj eta}\begin{split}\eta(X)&= \inf\{\phi(S)\mid S\in\mathcal S,\,X-S\in\CA_\rho\}\\
&= -\sup\{\phi(X-S)-\phi(X)\mid S\in\mathcal S,\,X-S\in\CA_\rho\}\\
&\ge\phi(X)-\rho^*(\phi).\end{split}\end{align}
That is, $\phi\in\dom(\eta^*)$. 
Next observe that, for nonconstant $S\in\mathcal S$, $x\in\dom(\eta)\cap \R$, and $t\in\R$, $\eta(x+tS)=\eta(x)+t\price(S)$. By $\P$-law invariance of $\eta$, \cite[Theorem 4.1]{Collapse} yields $\dom(\eta^*)\subset\R$. As $\phi(1)=1$, this entails $\dom(\eta^*)=\{1\}$, and \eqref{only trivial} is proved. 

As the $\P$-consistent risk measure $\rho$ is admissible, we can fix a nonconstant $Z\in\dom(\rho^*)$ (Proposition~\ref{compatible nonempty}). By Lemma~\ref{lem:structure}(2), we can assume that $Z$ only takes two values, say, $x<1<y$. Let $p:=\P(Z=x)$ and select a Hamel basis $\{S_1,\ldots,S_m\}$ of $\mathcal S$ satisfying $\E_\P[S_j]=1$. Let $k>0$ be large enough such that $S_j+k\in L^\infty_+$, $j\in[m]$ and  consider the vector-valued measure $\boldsymbol{\mu}\colon \CF\to[0,\infty)^{m+1}$ defined by 
$$\boldsymbol{\mu}(A)=\big(\P(A),\E_\P[(S_1+k)\ind_A],\ldots,\E_\P[(S_m+k)\ind_A]\big).$$ 
Each coordinate measure is an atomless measure. Lyapunov's Convexity Theorem \cite[Theorem 13.33]{Ali} implies that $\mathcal R:=\{\boldsymbol{\mu}(A)\mid A\in\CF\}$ is convex. In particular, 
$\big(p,p(k+1),\ldots,p(k+1)\big)\in \mathcal R$,
i.e., we find $B\in\CF$ such that $\P(B)=p$ and $\E_\P[(S_j+k)\ind_B]=p(k+1)$, $j\in[m]$. In particular, $\E_\P[S_j\ind_B]=p$ for all $j\in[m]$.
Now consider $\w Z:=x\ind_B+y\ind_{B^c}\sim_\P Z$, i.e., $\w Z\in\dom(\rho^*)$, and compute for all $j\in[m]$ that 
\[\E_\P[\w Z S_j]=x\E_\P[S_j\ind_{B}]+y\E_\P[S_j\ind_{B^c}]=xp+y(1-p)=1=\E_\P[S_j]=\price(S_j).\]
Consequently, $\E_\P[\w Z S]=\price(S)$ holds for all $S\in\mathcal S$, and we conclude as in \eqref{conj eta} that $\w Z\in\dom(\eta^*)$, contradicting \eqref{only trivial} established above. 
\end{proof}

\section{Examples}\label{sec:examples}

\begin{example}\label{ex:onlyone}
For $n\in\N$ define the convex monetary risk measure $\tau_n\colon L^\infty\to\R$ by
\[\tau_n(X):=\begin{cases}\es_1(X)&n=1,\\
\es_{\frac 1 n}(X)+n&n\ge 2,\end{cases}\]
where $\es_p$, $p\in[0,1]$, denotes the Expected Shortfall under $\P$ at level $p$; i.e., for a version of the quantile function $q_X$ of $X$ under $\P$, 
\begin{equation}\label{eq:ES}\es_p(X)=\begin{cases}\tfrac 1{1-p}\int_p^1q_X(s)\,{\rm d}s&p\in (0,1),\\
\sup_{s\in(0,1)}q_X(s)&p=1.
\end{cases}\end{equation}
The functional $\rho\colon L^\infty\to\R$ defined by  $\rho(X):=\inf_{n\in\N}\tau_n(X)$ is a consistent risk measure which is {\em not} star shaped (and therefore also not convex). 
Next we observe that assertion (5) in Proposition~\ref{compatible nonempty} applies. To see this, we may invoke that $(\Omega,\CF,\P)$ is atomless, select 
$A_k\in\CF$ satisfying $\P(A_k)=\frac 1 k$, and set $X_k:=-k^2\ind_{A_k}$, $k\in\N\setminus\{1\}$. Then $\tau_1(X_k)=0$ and $\tau_n(X_k)=n$ for all $2\le n\le k$. However, for $n>k$, we observe \[\tau_n(X_k)=-\tfrac{nk^2}{n-1}\cdot\left(\tfrac 1 k-\tfrac 1n\right)+n\ge n\cdot\left(1-\tfrac{k}{n-1}\right)\ge n\cdot\left(1-\tfrac{n-1}{n-1}\right)=0.\]
Hence, $\rho(X_k)=0$ for all $k\in\N$, while $\E_\P[X_k]=-k$.

Nevertheless, $\mbf C(\rho)=\varnothing$. It is well known that, for all $n\ge 2$,
$\dom(\tau_n^*)=\{Z\in L^1_+\mid\E_\P[Z]=1,\,Z\le \tfrac n{n-1}\}.$
This entails $\{1\}\subset\dom(\rho^*)\subset \bigcap_{n=2}^\infty\dom(\tau_n^*)=\{1\}$, whence $\mbf C(\rho)=\varnothing$ follows with Proposition~\ref{uncountably many} and Lemma~\ref{lem:star}.  
\end{example}

\begin{example}\label{ex:many}\
\begin{enumerate}[(1)]
\item It is well possible that $\mbf C(\rho)\subsetneq\dom(\rho^*)\cap L^1$. As an example, fix $0<p<1$ and suppose that $\rho=\es_{p}$, the Expected Shortfall at level $p$ defined as in \eqref{eq:ES}. 
$\rho$ is convex, normalised, and $\P$-consistent. 
We have recalled in the previous example that $\dom(\rho^*)=\{Z\in L^\infty_+\mid Z\le \tfrac{1}{1-p},~\E_\P[Z]=1\}.$
As $(\Omega,\CF,\P)$ is atomless, we may fix an event $A\in\CF$ with $\P(A)=1-p$. Let $Z^*:=\tfrac1{1-p}\ind_A$ and $U:=-\ind_{A^c}$. 
For all $t\ge 0$, $\rho(tU)\le\rho(0)=0$.
Moreover, for $Z\in\dom(\rho^*)$, $\E_\P[ZU]=0$ is equivalent to $Z\ind_{A^c}=0$ $\P$-a.s., whence we conclude that $Z=Z^*$. 
Hence, $U\in\CA^\infty_\rho$ and $\E_\P[Z^*U]=0$, but $U\neq 0$, which means that $Z^*\notin\mbf C(\rho)$.
\item Consider the entropic risk measure 
\[\rho(X):=\log\left(\E_\P[e^X]\right)=\sup_{Q\in L^1_+:\,\E_\P[Q]=1}\E_\P[QY]-\E_\P\left[Q\log(Q)\right],\quad X\in L^\infty.\]
$\rho$ is a normalised, convex, and $\P$-consistent risk measure which satisfies $\CA_\rho^\infty=-L^\infty_+$. 
In particular, each $Q\in L^1_+$ which satisfies $\E_\P[Q]=1$ and $\P(Q>0)=1$ lies in $\mbf C(\rho)$.
\item Examples (1)--(2) consider convex risk measures. From \cite[Example 3.3]{Consistent} we recall that a nonconvex  $\P$-consistent risk measure is given by considering the convex monetary risk measures 
\[\begin{array}{l}\tau_1(X)=\tfrac 1 2 \es_1(X)+\tfrac 1 2\E_\P[X],\\
\tau_2(X)=\es_{\frac 1 2}(X),\end{array}\quad\quad X\in L^\infty,\]
where $\es_p$ is defined as in \eqref{eq:ES}, and setting  $\rho:=\min\{\tau_1,\tau_2\}$.
Using, for instance, \cite[Theorem 5]{Starshape} or a direct verification, one shows that $\rho$ is star shaped. Moreover,  
\begin{align*}
&\dom(\tau_1^*)\cap L^1=\{Q\in L^1_+\mid\E_\P[Q]=1,\,Q\ge \tfrac 1 2\},\\ 
&\dom(\tau_2^*)\cap L^1=\{Q\in L^\infty_+\mid\E_\P[Q]=1,\,0\le Q\le 2\}.
\end{align*}
Hence, $\dom(\rho^*)\cap L^1=\{Q\in L^\infty_+\mid \E_\P[Q]=1,\,\tfrac 1 2 \le Q\le 2\}$, an observation that illustrates Proposition~\ref{uncountably many}.
\end{enumerate}
\end{example}

\begin{example}\label{ex:add1}
Fix an event $A\in\CF$ with $\P(A)=\tfrac 12$ and consider the probability measure $\Q\approx\P$ defined by 
$\tfrac{{\rm d}\Q}{{\rm d}\P}=\tfrac 1 2\ind_A+\tfrac 3 2\ind_{A^c}$. Denoting the Expected Shortfall under $\P$ by $\es^\P_\bullet$ and the Expected Shortfall under $\Q$ by $\es^\Q_{\bullet}$, we set 
$\rho_1:=\es_{1/5}^\P$ and $\rho_2:=\es_{1/6}^\Q$. 
One then identifies
$\dom(\rho_1^*)=\{Z\in L^\infty_+\mid \E_\P[Z]=1,\, Z\le 1.25\}$
and
$\dom(\rho_2^*)=\{Z\in L^\infty_+\mid \E_\P[Z]=1,\, Z\le 2.4\ind_A+0.8\ind_{A^c}\}.$
It is obvious that $\rho^*_2(1)=\infty$ and that $\rho_1^*(\tfrac{{\rm d}\Q}{{\rm d}\P})=\infty$. However, 
$\w Z:=1.22\ind_{A}+0.78\ind_{A^c}$
lies in $\mbf C(\rho_1)\cap\mbf C(\rho_2)$.
 Indeed, for $\lambda<1$, but close to 1 and suitably chosen, we can guarantee that $V:=\lambda^{-1}(Z-1+\lambda)$ satisfies $V\in L^\infty_+$ and $V\le 1.25$, i.e., $V\in\dom(\rho_1^*)$. In particular, $\w Z=\lambda V+1-\lambda$, which means $\w Z\in\mbf C(\rho_1)$ by \eqref{eq:is compatible}. One similarly verifies $\w Z\in\mbf C(\rho_2)$.
\end{example}

The following example illustrates the necessity of the assumptions of Theorem~\ref{thm:risksh_mon} in case $n=2$. 

\begin{example}\label{ex:big}\
\begin{enumerate}[(1)]
\item {\em All appearing reference probability measures have to be equivalent:} Suppose otherwise that $\Q\ll\P$ is arbitrarily chosen with the property $\Q\not\approx\P$. Consider the convex monetary risk measures $\rho_1(X):=\log\left(\E_\P[e^X]\right)$ and $\rho_2(X):=\E_\Q[X]$, $X\in\Linfty$. 
By \cite[Example 6.1]{JST}, $\rho_1\Box\rho_2$ is not exact at all $X\in L^\infty$. 

\item {\em Equivalence of reference probability measures alone is not enough:} To illustrate this, consider the convex monetary risk measures defined in (1) and assume $\Q\approx\P$. If $\rho_1\Box\rho_2$ were exact at each $X\in\Linfty$, one could follow the reasoning of \cite[Example 6.1]{JST} and conclude that, if $\rho_1\Box\rho_2(X)=\rho_1(X_1)+\rho_2(X_2)$ for some allocation
$\mbf X\in\mathbb A_X$, then $\tfrac{{\rm d}\Q}{{\rm d}\P}$ would have to be a subgradient of $\rho_1$ at $X_1$, i.e., $\rho_1(X_1)=\E_\P[\frac{{\rm d}\Q}{{\rm d}\P}X_1]-\rho_1^*\big(\tfrac{{\rm d}\Q}{{\rm d}\P}\big)$.
By \cite[Lemma 6.1]{Subgradients}, the identity 
\[\frac{{\rm d}\Q}{{\rm d}\P}=\frac{e^{X_1}}{\E_\P[e^{X_1}]}\]
must hold. As $X_1\in\Linfty$, $\log(Q)$ has to be bounded from above and below, that is, there have to be constants $0<s<S$ such that $\P(Q\in[s,S])=1$.

\item {\em Assumption (iii) cannot be dropped, even if Assumption~\ref{assum1} holds:} In case $n=2$, assumption (iii) of Theorem~\ref{thm:risksh_mon} reads as $\mbf C(\rho_1)\cap \dom(\rho_2^*)\neq\varnothing$.
Here we use \cite[Example 4.3]{Acciaio}.
Consider the probability measure $\Q\approx\P$ from Example~\ref{ex:add1} and define the convex monetary risk measures
\[\begin{array}{l}\rho_1(X):=\tfrac 1 2\E_\P[X]+\tfrac 1 2\log(\E_\P[e^{X}]),\\
\rho_2(X):=\E_{\Q}[X],\end{array}\quad\quad X\in L^\infty.\]
By \cite[Example 4.3]{Acciaio}, $\rho_1\Box\rho_2$ is not exact at all $X\in L^\infty$. However, as $\dom(\rho^*_2)$ is a singleton, $\mbf C(\rho_1)\cap \dom(\rho_2^*)\neq\varnothing$ would hold if and only if $\tfrac{{\rm d}\Q}{{\rm d}\P}\in\mbf C(\rho_1)$. Setting $A=\{\tfrac{{\rm d}\Q}{{\rm d}\P}=\tfrac 1 2\}$, we consider $U:=-\ind_A+\tfrac 1 3\ind_{A^c}$ and show by a direct computation that $U\in\CA_{\rho_1}^\infty$. As, however, $\E_{\Q}[U]=0$, $\tfrac{{\rm d}\Q}{{\rm d}\P}$ cannot be compatible for $\rho_1$. 
\end{enumerate}
\end{example}

\begin{example}\label{ex:add2}
Consider the homogeneous case $\Q_1=\Q_2=\P$. 
Moreover, let $\rho_1=\rho_2$ be two entropic risk measures on $L^\infty$ defined by 
$\rho_i(X)=\tfrac 1 2\log\big(\E_\P[e^{2X}]\big)$, $i=1,2$. By \cite[Example 2.9]{Subgradients}, $\rho:=\rho_1\Box\rho_2$ is given by $\rho(X)=\log(\E_\P[e^X])$. Suppose $Z^*\in L^1_+$ satisfies $\E_\P[Z^*]=1$ and $\E_\P[Z^*\ind_A]=0$ for some $A\in\CF$ with $\P(A)>0$. 
Define capital requirement $\eta_1$ in terms of the risk measurement regime $\CA_{\rho_1},\tn{span}\{1\},id_\R)$, i.e., $\rho_1=\eta_1$. For $\eta_2$, we consider the security space $\mathcal S_2:=\tn{span}\{\ind_\Omega,\ind_{A^c}\}$ and assume $\price_2(S)=\E_\P[Z^*S]$, $S\in\mathcal S_2$. 

As the global security space is given by $\mathcal M=\mathcal S_2$, we also have 
$\ker(\pi)=\R\cdot\ind_A$. 
The associated risk sharing functional $\eta$ is easily seen to be finite.
Assume for contradiction that $\eta$ is exact at all $X\in\dom(\eta)=L^\infty$ and fix an optimal allocation $\mbf X\in\mathbb A_X$. 
As $\eta$ is convex and continuous at $X$, we can find some subgradient $Z$ at $X$. This subgradient has to satisfy $Z\in L^1_+$ because $\eta$ has the Lebesgue property, i.e., $\dom(\eta^*)\subset L^1$. We obtain
\[\eta(X)=\E_\P[ZX]-\eta^*(Z)=\sum_{i=1}^2\E_\P[ZX_i]-\rho_i^*(Z)\le \sum_{i=1}^2\rho_i(X_i)=\eta(X).\]
Hence, $Z$ is a subgradient of $\rho_i$ at $X_i$, $i=1,2$. Applying \cite[Theorem 3]{FKM2015}, $\E_\P[Z\boldsymbol{\cdot}]|_{\ker(\pi)}\equiv 0$, or equivalently, 
$\E_\P[Z\ind_A]=0$. However, by \cite[Lemma 6.1]{Subgradients},
$Z=\E_\P[e^{X_2}]^{-1}e^{X_2}.$
No $Z$ of this shape can satisfy $\E_\P[Z\ind_A]=0$. 
\end{example}


\end{document}